\documentclass{article}
\usepackage[a4paper,margin=1in]{geometry}
\usepackage{amsmath,amssymb,amsfonts}
\usepackage{graphicx}
\usepackage{hyperref}
\usepackage{MnSymbol}
\usepackage{algorithm}
\usepackage{algorithmic}
\usepackage{stmaryrd}
\usepackage{tabularx}
\usepackage{longtable}
\usepackage{amsthm}

\newtheorem{thm}{Theorem}
\newtheorem{cor}[thm]{Corollary}

\newtheorem{lm}[thm]{Lemma}
\newtheorem{ex}{Example}
\hypersetup{colorlinks=true, linkcolor=blue, citecolor=blue}

\title{A Truthful Multiunit Profit-Optimal Mechanism for Synthesizing Social Laws\thanks{This is the full and extended version of the conference paper appearing in Proceedings of the 35th International Joint Conference on Artificial Intelligence (IJCAI-26), 2026.}
	}
\author{
	Jun Wu\\
	Nanjing University of Finance and Economics, China \\
	\href{mailto:jwucie@nufe.edu.cn}{jwucie@nufe.edu.cn}
	\and
	Jian Huang\thanks{Corresponding author.}\\
	Nanjing University of Finance and Economics, China \\
	\href{mailto:jianhuangvictor@gmail.com}{jianhuangvictor@gmail.com}\\
    Chongjun Wang\\
	Nanjing University, China \\
	\href{mailto:chjwang@nju.edu.cn}{chjwang@nju.edu.cn}\\
}
\date{} % 不要日期

\begin{document}
	\maketitle
	
	\begin{abstract}
    This paper studies Social Law Synthesis (SLS) in strategic multi-agent environments as a new multi-unit mechanism design problem. We model SLS as a Bayesian single-parameter procurement auction based on Alternating-time Temporal Logic (ATL) and aim to design a truthful, individually rational, and profit-optimal mechanism. We first prove a representation lemma showing that any valuation respecting alternating bisimulation can be compactly expressed as a feature set of ATL formulae. We then reduce payment determination to allocation determination in polynomial time, resolving the irregular payment issue inherent in multi-unit settings. We further show that allocation determination is \(FP^{NP}\)-complete and encode ATL semantics into integer linear programming (ILP) constraints to make the problem tractable with standard solvers. Based on these results, we present the $\mathcal{PO\text{-}ASL}$ mechanism, which is incentive-compatible, individually rational, and maximizes expected profit. Theoretical guarantees and examples confirm that our approach provides an effective and computationally feasible solution for synthesizing optimal social laws under strategic agent behavior.
	\end{abstract}
	
	\section{Introduction}
	Social laws have proven to be a powerful and theoretically elegant approach to coordinating multi-agent systems. Based on modal logic, we can model a multi-agent system as a semantic structure, specify the system’s properties as logical formulas, and then verify these formulas via model checking. A social law is a set of restrictions on the available actions of agents, altering the system’s underlying structure in the hope that desirable properties (objectives) will emerge in the new structure~\cite{Agotnes:AAMAS-07b}. By enabling the expression of valuations for potential social laws and accounting for the cost of implementing each action restriction, SLS can be modeled as an optimization problem ~\cite{Agotnes:AAMAS-10}. However, we find that information incompleteness and the self-interested nature of agents, common in typical strategic environments, can introduce unavoidable obstacles to solving this problem. For example, the costs of action restrictions, which are key optimization parameters, are known only privately to each agent; if we attempt to elicit these costs by “asking” the agents, they may strategically misreport them to maximize their own utility.
	%, or even simply refuse to provide any answer and disobey the selected social laws in the future.  
	
	This challenge naturally falls into the domain of algorithmic mechanism design~\cite{Nisan07}. We adopt the powerful Alternating-time Temporal Logic (ATL)~\cite{Alur:JACM-02} as our foundation, enabling us to model SLS as an auction design problem characterized by the following features:
	%\vspace*{-1mm}
	\begin{itemize}
		\item [1)] \emph{Multi-unit Procurement}: Each agent has multiple available actions, and the social law designer seeks to selectively restrict a subset of actions for each agent.
		\vspace*{-1mm}
		\item [2)] \emph{Bayesian single-parameter setting} : Each agent holds private information about the unit cost of restricting each of its actions, while the prior probability distribution of these unit costs is common knowledge.
		\vspace*{-0.5mm}
		\item [3)] \emph{Profit-Optimality}: The objective is to reliably output the social law that maximizes profit (defined as value minus payment), even when agents behave strategically.
	\end{itemize}
	%\vspace*{-0.5mm}
	The primary challenge stems from the “multi-unit” nature of the problem, which prevents us from deriving a clean solution directly from Myerson’s framework~\cite{Myerson:MOR-81,Elkind:SODA-04}. Indeed, the interplay between the “multi-unit” setting and the sophisticated semantics of ATL renders both the modeling and computation of the mechanism beyond the reach of current state-of-the-art techniques. We aim to develop a solution by carefully analyzing and leveraging ATL’s syntax and semantics.
	
	Our main contributions are summarized as follows:
	\vspace*{-1mm}
	\begin{itemize}
		\item[1)] We formally prove a representation lemma for ATL-based social laws: any valuation function that does not distinguish alternating bisimulation-equivalent structures~\cite{Alur:ICCT-98} can be succinctly represented as a feature set, allowing us to formalize SLS as an auction design problem.
		
		\item[2)] We identify that the critical obstacle from “multi-unit” settings is irregular payment calculation. We address this by designing an efficient polynomial-time reduction of payment determination to allocation determination, making allocation the key problem to solve.
		
		\item[3)] We finally prove that allocation determination is FP$^{NP}$-complete, and then reduce it to ILP in polynomial time by successfully encoding ATL semantics into a set of ILP constraints. This offloads the computational intractability to highly mature ILP solvers, which are widely used for large-scale industrial problems. Since payment determination reduces to allocation determination in polynomial time, we obtain a mechanism that is truthful, individually rational, profit-maximizing, and efficiently computable with ILP solvers.
		
	\end{itemize}
	
	Building on the above work, we propose a self-enforcing social law synthesis method that reliably optimizes expected profit in strategic environments.
	
	This paper can be seen as an attempt to introduce the methodology of algorithmic mechanism design into the traditional logic-based approach to artificial intelligence. We have obtained a framework that not only greatly improves the reliability and robustness of social laws, but also enriches the study of algorithmic mechanism design.
	
	The remainder of this paper is structured as follows: We start with some backgrounds on social laws and algorithmic mechanism design as well as the formal framework of our work. Next, we present an optimal social law auction and prove it maximizes the expected profit within all Bayesian-Nash incentive compatible and individual rational mechanisms, and try to figure out an efficient approach to compute the payment. Then, we prove that the proposed mechanism is computationally intractable and propose a ILP based algorithm for computing the proposed mechanism.  Finally, we present some conclusions and introduce some interesting open problems for future study.
	
	\section{Related Work}
	Our work in this paper can be seen as to solve optimal social law synthesizing problem in the strategic case, where the agents are rational in the sense of game theory and hold some private parameters, based on the methodology of Bayesian mechanism design~\cite{chen:AIJ-22,Elkind:SODA-04,Myerson:MOR-81}, where prior distributions of the private parameters are public information. So, our work mainly relates to social laws and algorithmic mechanism design.
	
	Since different logic systems can provide different semantic structures or symbolic languages, which can respectively model different kinds of mutiagent systems or specify different coordination objectives, trying to use different logics for multiagent systems and corresponding model checking technologies to implement social laws is an emphasized research subject in this area. For example, using simpler CTL language to specify coordination objectives to simplify the theoretical framework of van der Hoek~\cite{van_der_Hoek:Synthese-07} can obtain some further theoretical results~\cite{Agotnes:AAMAS-07b,Agotnes:IJCAI-07b,Agotnes:AAMAS-08a,Agotnes:AAMAS-09}. Using Alternating-time Temporal Epistemic Logic (ATEL) to specify coordination objectives can enable the agents to coordinate on ``knowledge"~\cite{van_der_Hoek:AAMAS-05b}; By specifically designing a new logic named Co-ATL, which put restrictions on the joint actions of agent coalitions instead of the actions of individual agents, can implement a social law which can modify the system more flexibly~\cite{Wu:AAMAS-11,Wang:TSMCB-11}.
	
	How to make autonomous agents consciously obey the constraints of social laws is also a basic issue that has been intensively discussed in this field. Since Shoham and Tennenholtz~\cite{Shoham:AAAI-92,shoham:AIJ-95} proposed social laws, ``all agents will unconditionally obey the selected social laws " became a basic assumption in this field, and how to relax this assumption has been an urgent problem to be solved. ``Why a rational and self-interested agent chooses to obey a social law when knowing that it will obtain lower returns" has been one of the paradoxes involved in the study of social laws. Therefore, ``What is the driving mechanism for the agent to comply with a social law?", ``What will happen if some agents do not comply with the social law?", and ``How robust is the system against agents that do not comply with the social law?" are basic questions that should be answered ~\cite{van_der_Hoek:Synthese-07}. In response to these problems, Binmore~\cite{Binmore-94,Binmore-98} studied social laws from the perspective of game theory, trying to explain why social laws can exist in a multi-agent system composed of self-interested agents; After introducing logics to re-formalize social laws, a series of studies on ``compliance" conducted by T. \AA gotnes et al. \cite{Agotnes:AAMAS-07b,Agotnes:AAMAS-08a,Agotnes:AAMAS-09} are important work in this area. 
	
	The cost of implementing a social law was first studied by ~\cite{Fitoussi:AIJ-00}.  ~\cite{Agotnes:AAMAS-10} later proposed the {\sc Optimal Social Law} problem, which took into account both the cost and benefit of implementing a social law and models computational tree logic (CTL)~\cite{Emerson90,Clarke00}-based social law synthesizing as a combinatorial optimization problem.  \cite{Wu:AAMAS-17a,Wu:JAAMAS-21} extended the {\sc Optimal Social Law} problem to the strategic case, developed solutions based on the framework of optimal auctions~\cite{Myerson:MOR-81,Elkind:SODA-04}. The problem considered in this paper is mainly motivated from this series of work. Compared with the previous work, we consider more general ATL-based social laws, and it turns out to be a brand new \emph{multiunit auction} design problem~\cite{Maskin:EMIG-89,Bhattacharya:TCS-20,Simina:AIJ-23}. Note that, rational behavior of the agents in the sense of game theory actually has already been considered in some work on social laws, e.g., \cite{Agotnes:AAMAS-07b} and so on, but as far as we know they mainly focus on game theoretical analysis instead of on synthesizing which reduces to mechanism design, as we do in this paper. Bulling and Dastani~\cite{Bulling:IJCAI-11,Bulling:AIJ-16,Dell:BNAIC-17} introduced the insightful methodology of normative mechanism design, which is at first glance very similar to ours. But in fact there are fundamental differences especially in the motivation. We mainly aim to solve the optimal social law synthesizing problem where the agents are rational and have some private parameters under the framework of algorithmic mechanism design, while the concept of mechanism design in Bulling and Dastani's framework is different from the traditional situation, \emph{i.e.},  it assumes that the preferences of each agent are public information, and there is actually no concept of private parameters. Another difference is that this paper have derived an algorithm based on integer programming via analyzing and utilizing the characteristics of ATL semantics, and strictly proved its correctness, while most of the existing related work ends with the discussion of computational complexity~\cite{van_der_Hoek:Synthese-07}. 
	
	Game Theory and Mechanism Design are traditional research directions of Economics. In short, game is a mathematical model for strategical interactions of rational individuals; Game Theory provides a set of method for analyzing games; mechanism design provides a set of methods for designing games. Under the effect of some fundamental theorems such as the revelation principle~\cite{Myerson:MOR-81}, mechanism design is often equivalent to ``Truthful Auction Design", that is, to design auction rules that incentivize bidders to truthfully report their private information. The original motivation of the algorithmic mechanism design research is to integrate the concepts and methodologies of mechanism design into the theoretical framework of traditional algorithms (thereby obtaining an ``algorithmic mechanism"), so that it can solve the problem and at the same time provide correct incentives, and therefore applicable to the related optimization problems of multi-agent systems composed of rational individuals~\cite{Nisan:GEB-01,Nisan:STOC-99,Nisan07}. The work of this paper is an attempt to put the above research ideas into practice in the problem of social laws synthesizing, and further expands the application field of algorithmic mechanism design. Optimization objective is an important dimension to distinguish the problem of mechanism design. ``Efficiency mechanism design problems" with the goal of maximizing social welfare, such as path auctions~\cite{Archer:SODA-02,Talwar:STACS-03,Archer:TOA-04,Elkind:SODA-04,Calinescu:TCS-15,zhang:ECAI-16}, minimum spanning tree auctions~\cite{Bikhchandani-02,Garg:TR-02,Talwar:STACS-03}, and combinatorial auctions which derived from the spectrum auctions of FCC since the 1990s~\cite{Sandholm:AIJ-02}, etc., can all be solved by the well-known Vickrey-Clarke-Groves (VCG) mechanism~\cite{Vickrey-61,Clarke-71,Groves-73}. Although this mechanism does not have polynomial time implementation~\cite{Nisan:JAIR-07}, sometimes it is possible to design computationally feasible approximation mechanisms to obtain social welfare close to the global optimum~\cite{Nisan:JAIR-07,Dobzinski:STOC-07}. ``Optimal mechanism design" in the narrow sense refers to mechanism design problems with the optimization goal of revenue maximization~\cite{Myerson:MOR-81}. In a broad sense, it can refer to any mechanism design problems with optimization goals or constraints related to the auctioneer’s payment~\cite{Hartline06,Wu:AAAI-20,Wu:AAMAS-17a,Wu:AAMAS-19}.
	
	On the side of algorithm mechanism design, our work especially relates to the class of work on graphs, e.g., path auctions, spanning-tree auctions, auction design on social networks~\cite{Li:AIJ:2022,Meng:GEB-2022,Li:AAAI-17}, and so on, since social law synthesis is a problem on Kripke structures, which is also intrinsically graphs. Among this class of work, our work is most close to Elkind et al.'s work~\cite{Elkind:SODA-04} on path auctions, since it proposes a methodology that solves the payment minimization problem in path auctions in the Bayesian case based on the framework of Myerson's optimal auctions~\cite{Myerson:MOR-81}. But our work is different from~\cite{Elkind:SODA-04} on at least two aspects: firstly, we study a different optimization objective, \emph{i.e.}, $(value-payment)$ vs. $payment$; secondly, we study a multiunit setting instead of a singe-unit setting. Actually, profit-maximization has been studied in \emph{competitive auctions}~\cite{Goldberg:SODA-01,Goldberg:GEB-06,Chen:STOC-14,Cary:SODA-08}, but their settings are totally different, \emph{i.e.}, they focus on buying (or selling) multiple identical items in the prior-free case where no bidder's private parameter distribution is available. Our problem also belongs to the class of multi-unit mechanism design problems~\cite{Maskin:EMIG-89,Bhattacharya:TCS-20,Simina:AIJ-23,Hagen2023Collusion,Gautier2023Multi,Potfer2024Improved,Bougt2025Revenue,Aberg2025Quantifying}, especially multi-unit procurement auction design problem~\cite{chan2014truthful,Wu:AAMAS-19,Wu:AAAI-20}, where the bidders possess multiple units of an item and the auctioneer buy a certain amount of items from each agent to optimize a given objective. It is well known that the case of multi-unit can usually create some major challenges to both allocation and payment determination.
	
	\section{Preliminaries}
	
	In this section, we introduce our problem setting. An illustrative example is provided in the appendix (\emph{cf.} Example 1).
	
	\subsection{Optimal social law synthesizing problem}
	Social laws can be formalized based on ATL, which adopts Concurrent Game Structures ({\sc cgs}) as the semantic structure. 
	
	\vspace*{1mm}
	\noindent\textbf{Cost-aware concurrent game structure, CCGS:} % Concurrent game structures, which is evolved from Kripke structures~\cite{Emerson90}, is a universal model for open systems. 
	We define {\sc ccgs} as an extended {\sc cgs} where a marginal cost function which specifies each agent's unit cost is added. Formally, a {\sc ccgs} is a tuple $S=\langle k, Q,  q_s, \Pi, \pi, \varepsilon, c, \delta \rangle$ with:
	\begin{itemize}
		
		\item A set $Ag=\{1,...,k\}$ of agents, a finite set $Q$ of states, an initial state $q_s\in Q$, a finite set $\Pi$ of propositions, a labeling function $\pi$ specifying for each state $q \in Q$, a set $\pi(q) \subseteq \Pi$ of true propositions, and an action function $\varepsilon$. For each agent $a\in Ag$ and each state $q\in Q$, a non-empty set $\varepsilon_a(q)$ of actions is available to agent $a$ at state $q$. A \emph{joint action} of all the agents at state $q$ is a tuple $\langle j_1,...,j_k \rangle$ such that $j_a \in \varepsilon_a(q)$ for each agent $a$. We write $D(q)$ for the joint action space $\varepsilon_1(q)\times ... \times \varepsilon_k(q)$.
		
		\item A unit cost function $c$. For each agent $i\in Ag$, if $n\in \mathbb{N}$ of her actions are restricted, then the resulted cost to her is $n\cdot c(i)$. We assume $c(i)$ to be agent $i$'s private information, but as public information $c(1),...,c(k)$ are independent continuous random variables $x_1,...,x_k$ respectively, where each $x_i$ is drawn from the interval $X_i = [0,\omega_i]$ subject to a probability density function $f_i$.
		
		\item A transition function $\delta$. For each state $q\in Q$ and each joint action $\langle j_1,...,j_k \rangle\in D(q)$, $\delta(q, j_1,...,j_k)\in Q$ is the next state if every agent $a\in Ag$ chooses action $j_a$. 
		
	\end{itemize}
	
	A {\sc ccgs} intrinsically specifies the running of a multiagent system: each agent chooses an action simultaneously, and the actions chosen by all the agents determine the next state.   
	%So, for a coalition formed by part of the agents, it can only partially determine the system transition path, due to the uncertainty of strategy selection of other agents. %The logic {\sc atl} is designed to specify and verify this kind of strategic reliability of coalitions.
	
	Note that, our cost model contains an assumption of \emph{constant marginal cost per action} restriction. It actually tries to generalize the “\emph{minimality}” objective in SLS proposed by~\cite{Fitoussi:AIJ-00}, which requires
	the social law to restrict as few actions as possible, and thus include least “side effects”. Our cost model makes sense especially in cases when it is hard for the agents to determine the actual cost for restricting each individual action. Actually, it is easier for them to specify a rough or average cost per action restriction, and in this way we can well capture the agents’ preference on fewer restrictions. We will study more general cost models in our future work.
	
	%Frankly, allowing the agents to specify different costs on different actions must be a more reasonable choice. However, it will result in a multi-parameter optimal mechanism design problem, which is still an open problem in economics. We hope that we can further study these in the future.
	
	\vspace*{1mm}
	\noindent\textbf{Alternating-time Temporal logic, ATL:} We can adopt {\sc atl} for specifying and verifying CCGSs. 
	The language of {\sc atl}, denoted as $\mathcal{L}_{ATL}$, is generated by the
	following grammar:$$\varphi ::= p|\neg
	\varphi|\varphi_1\vee\varphi_2|\llangle
	A\rrangle\bigcirc\varphi|\llangle A\rrangle\Box\varphi|\llangle
	A\rrangle\varphi_1\mathcal{U}\varphi_2,$$ where $p\in\Pi$ is a
	proposition, and $A\subseteq Ag$ is a coalition. 
	%We assume that we are studying a fixed set $Ag$ of agents and a fixed set $\Pi$ of propositions.

	The symbols $\bigcirc, \Box, \mathcal{U}$ are the conventional linear-time temporal operators ``\textit{next-time}", ``\textit{always}", ``\textit{until}", respectively. For a set $A\subseteq Ag$, the operator $\llangle A\rrangle$ is a path selection operator, with  $\llangle A\rrangle\psi$ aims to denote ``\textit{coalition $A$ can reliably enforce property $\psi$"}. Given a {\sc atl} formula $\varphi$, we use the notation $S, q\vDash \varphi$ to mean ``\emph{$\varphi$ is satisfied in the state $q$ of $S$}". 
	
	\vspace*{1mm}
	\noindent\textbf{Social laws:} A social law is an action constraint $\eta$  for a {\sc ccgs} $S=\langle k, Q,  q_s, \Pi, \pi, \varepsilon, c, \delta \rangle$ defined as a function $\forall i\in Ag, q\in Q: \eta_i(q)\subset \varepsilon_i(q)$.
	The new structure obtained by \emph{implementing} social law $\eta$ on $S$, denoted $S\dag\eta$, is the structure $S' = \langle k, Q, q_s, \Pi, \pi, \varepsilon', c, \delta'\rangle$ where $\forall i\in Ag, q\in Q: \varepsilon'_i(q) = \varepsilon_i(q)\setminus \eta_i(q)$, and $\delta'$ is obtained from $\delta$ by restricting the domain of definition according to $\varepsilon'$. Intuitively, the implementing of a social law on a {\sc cgs} is deleting from it all the  actions restricted by this social law. 
	%We require the actions restricted for every agent to be a proper subset to guarantee nonempty available action sets after implementing social laws. 
	We let $\mathcal{SL}_S$ denote the set of all the possible social laws for $S$, and let $\mathcal{W}_S$ denote the set of possible {\sc ccgs}s that can be obtained from $S$ by implementing a social law, that is, $\mathcal{W}_S = \{S'~|~\exists \eta\in\mathcal{SL}_S: S\dag \eta = S'\}. $
	%So, actually $\mathcal{W}_S$ is all the possible results (structures) that can be obtained via implementing social laws, and 
	
	%Example 1 (in the appendix) shows that managing a system consisting of two agents, who share a data storage, can be modeled as a Social Law Synthesizing (SLS) problem.
	
	\vspace*{1mm}
	\noindent\textbf{Demand valuation function, DVF:} Each of the structures in $\mathcal{W}_S$ has a possibly different value, which can be specified as a {\sc dvf} $v: \mathcal{W}_S\rightarrow \mathbb{R}^+.$ %In fact, the {\sc dvf} defines a preference of the designer on the structures that can be obtained.
	From the perspective of each agent $i$, implementing a social law $\eta$ results in a cost $C_i=c(i)\cdot\sum_{q\in Q} |\eta_i(q)|$. If the value of $c(i)$ was known, we can pay exactly $C_i$ units of money to agent $i$ to compensate its cost. Therefore the overall payment to all the agent should be $\sum_{i\in Ag} C_i$, and the profit of implementing social law $\eta$ will be
	\begin{equation}
		g(\eta) = v(S\dag\eta) - \sum_{i\in Ag} (c(i)\cdot\sum_{q\in Q} |\eta_i(q)|)
	\end{equation}
	Finally, profit optimal social law synthesizing can be modeled as the optimization problem of finding a social law that maximizes the function $g(\eta)$. 
	%But now notice that, $c(i)$ is only known by agent $i$ herself, what we can obtain is only $c(i)$'s probability density function $f_i$. In such a setting, it is obviously impossible for us to directly optimize $g(\eta)$. How to select the social law and pay the agents to obtain the best profit and at the same time keep the agents well-motivated becomes a challenge. We will try to find out a solution based on algorithmic mechanism design.
	
	\subsection{Compactly represent the DVF}
	Based on the work in\cite{Alur:ICCT-98} and~\cite{Agotnes:AAMAS-10}, we can define a bisimulation relation between the state sets of two arbitrary structures in $\mathcal{W}_S$.
	
	\vspace*{1mm}
	\noindent\textbf{Alternating bisimulation relation:} 
	Given two cost-aware concurrent game structures (CCGSs) $S=\langle k, Q,  q_s, \Pi, \pi, \varepsilon, c,
	\delta \rangle$ and $S'=\langle k', Q',  q'_s, \Pi', \pi', \varepsilon', c',
	\delta' \rangle$ based on the same agent set $Ag = \{1,...,k\}$, any binary relation $Z\subseteq Q\times Q'$ is called an alternating bisimulation relation if $q_s Z q'_s$ and for any $q Z q'$, we have
	
	\begin{itemize}
		\item[1)] $\pi(q) = \pi'(q')$;
		
		\item[2)] For each $A\subseteq Ag$ and every $\vec{m}_A \in D_A(q)$, there is a $\vec{m}'_A \in D'_A(q')$ satisfying for all $q'_x\in out(q', \vec{m}'_A)$, there is a $q_x\in out(q,\vec{m}_A)$ satisfying $q_x Z q'_x$; and
		
		\item[3)] For each $A\subseteq Ag$ and every $\vec{m}'_A \in D'_A(q')$, there is a $\vec{m}_A \in D_A(q)$ satisfying for all $q_x\in out(q, \vec{m}_A)$, there is a $q'_x\in out(q',\vec{m}'_A)$ satisfying $q_x Z q'_x$.
	\end{itemize}

	%Intuitively, bisimulation relations represent a very nature equivalence concept. 
	
	Note that, according to ATL semantics, in any state $q\in Q$, $D_A(q)$ denotes the joint action space of coalition $A$,  and $out(q, \vec{m}_A)$ denotes the set of possible next states when coalition $A$ chooses the joint action $\vec{m}_A$. The readers can consult the appendix or literature~\cite{Alur:JACM-02} for more details about ATL syntax and semantics. We say two structures $S$ and $S'$ are \emph{alternating bisimulation equivalent}, denoted as $S \leftrightarroweq S'$, if there is a bisimulation relation between $S$ and $S'$. 
	
	\vspace*{1mm}
	\noindent\textbf{The representation lemma:} Based on the following result, we can formally prove the subsequent representation lemma.
	
	\begin{thm}\cite{Alur:ICCT-98} For two concurrent game structures $S$ and $S'$, if $S \leftrightarroweq S'$ and $q Z q'$, then for an arbitrary ATL function $\varphi$, we have $S, q\vDash \varphi~~\text{iff}~~S', q'\vDash \varphi.$
	\end{thm}
	
	\begin{lm}[Representation Lemma]\label{lm-rl} The following two statements are equivalent:
		\begin{itemize}
			\item[1)] $\forall S_1, S_2$: $S \leftrightarroweq S'$ implies $v(S) = v(S')$;
			\item[2)] $\exists$ feature set $\mathcal{F} = \{(\varphi_1,c_1),...,(\varphi_k,c_k)\}$, where $k\in\mathbb{N},\forall j\in Ag: \varphi_j\in\mathcal{L}_{ATL}, c_j\in\mathbb{R}^+$, and for all $S\in\mathcal{W}_S$:
			$$v(S)= \sigma(\mathcal{F},S) = \sum_{(\varphi_j,c_j)\in\mathcal{F}; S,q_s\vDash\varphi_j}c_j.$$
		\end{itemize}
	\end{lm}

	The above results imply that any {\sc dvf}s that do not distinguish between alternating bisimulation-equivalent structures can be equivalently and compactly represented as feature sets. %The requirement for the indiscrimination on alternating bisimulation equivalent structures is essentially rooted in the limitation of the expressive power of the ATL language. However, 
	We will focus on this type of {\sc dvf}s since it well reflects mathematical equivalence and is suitable for most scenarios. Note that, the proofs for our results can be found in the appendix.

	%Social laws are always synthesized to fulfill a predefined objective. The most basic objective can be denoted by the tuple $\langle q,\varphi\rangle$, meaning to make the ATL formula $\varphi$ satisfied in state $q$ of the new structure. In the naive non-strategic case, social laws are realized as hard constraints imposed on the agents, so we can simply find out an effective $\eta$, which satisfy $S\dag\eta, q\vDash\varphi$, from its space by searching. van der Hoek et al.~\cite{van_der_Hoek:Synthese-07} has systematically studied this case, and has established that verifying the effectiveness of a social law is equivalent to the ATL model checking problem and so is PTime-complete, and checking whether there is an effective social law is NP-complete. Moreover, the results in ~\cite{Wu:AAMAS-11} show the problem of synthesizing an effective social law is FP$^{NP}$-complete. 
	
\section{Mechanism Design}
\subsection{Social law Synthesis via auctions}

We select the social law by running a procurement auction: firstly we announce a mechanism consisting of an \emph{allocation function} $\hat{R}:X\rightarrow \mathcal{SL}_S$ and, for each agent $i$, a \emph{payment function} $P_i:X\rightarrow \mathbb{R}^+$. We then collect the bids (\emph{i.e.}, cost reports) from the agents to obtain the bid profile $x$, and finally select the social law $\eta=\hat{R}(x)$ and pay  $P_i(x)$ to each agent $i$. Note that, the function  $\hat{R}$ can be equivalently specified as $\forall i,q,a: R_{i:a}^q(x)= 1$ (indicating that agent $i$'s action $a$ in state $q$ is selected) if $a\in\eta(i,q)$,  $R_{i:a}^q(x) = 0$ otherwise. Now, we let $R_i(x)$ be the total number of agent $i$'s restricted actions when the bid profile is $x$, i.e.,
\begin{equation}
	R_i(x) = \sum_{q\in Q}\sum_{a\in\epsilon_i(q)}R^q_{i:a}(x)
\end{equation}

%\vspace*{1mm}
\noindent\textbf{Bayesian games:} After a mechanism $\langle R,P\rangle$ is determined, it applies to any instances of agent set $A$ where real costs are randomly drawn from  $X$, the system intrinsically becomes a Bayesian game $\langle A, (X_i, f_i, B_i, u_i)_{i\in A}\rangle$ where 
\begin{itemize}
	\item $X_i$ now denotes not only the cost space but also the action space for agent $i$, \emph{i.e.}, each $x_i\in X_i$ also denote the action ``bidding/reporting $x_i$";
	\item $B_i$ is the strategy space of agent $i$, consisting of all functions of the form $b:X_i\rightarrow X_i$, which allows agent i to report its cost strategically.
	\item $u_i:X_i\times X\rightarrow \mathbb{R}$ is agent $i$'s utility function. When agent $i$'s true cost is $x_i$,  a bid profile $x$ yields the utility: \begin{equation}\label{eq-uid}
		u_i(x_i,x) = P_i(x) - x_i\cdot R_i(x)
	\end{equation}
\end{itemize}
The expected utility achieved by agent $i$ whose unit cost is $x_i\in X_i$ when submitting a bid $x'_i\in X_i$ is
\begin{equation}\label{eq-baruid}
	\bar{u}_i(x_i,x'_i) = \mathbb{E}_{x_{-i}\in X_{-i}}[u_i(x_i,(x'_i,x_{-i}))]
\end{equation} 
%\vspace*{1mm}
\noindent\textbf{Bayesian-Nash Equilibrium (BNE):} A strategy profile $(b_1,\cdots, b_n)$ is a BNE if and only if
\begin{equation}
	\forall i, x_i, b'_{i}\neq b_{i}: \bar{u}_i(x_i,b(x_i))\geq \bar{u}_i(x_i,b'(x_i))
\end{equation}
A mechanism is called \emph{Bayesian-Nash Incentive Compatible (BNIC)} if and only if truthful bidding(\emph{i.e.}, $b_i(x_i) = x_i$ for all $i$ and $x_i$) constitutes a BNE. According to the famous \emph{Revelation Principle}~\cite{Nisan:GEB-01}, we can restrict the search space to BNIC mechanisms without loss of generality. 

When agent $i$ bids $x'_i$, we let $r^q_{i:a}(x'_i)$, $r_i(x'_i)$ and $p_i(x'_i)$ be the \emph{probability of agent $i$'s action $a$ in state $q$ being selected}, the \emph{expected number of agent $i$'s selected actions} and the \emph{expected amount of payment to agent $i$}, respectively, then 
\begin{equation}\label{eq-hi}
	r^q_{i:a}(x'_i) = \int_{X_{-i}}R^q_{i:a}(x'_i,x_{-i})f_{-i}(x_{-i})dx_{-i}
\end{equation}
\vspace{-2mm}
\begin{equation}\label{eq-hix}
	r_i(x'_i)\! =\!\!\! \sum_{q\in Q}\!\sum_{a\in\epsilon_i(q)}\!\!\!r^q_{i:a}(x'_i)\! =\!\!\! \int_{X_{-i}}\!\!\!\!\!R_i(x'_i,x_{-i})f_{-i}(x_{-i})dx_{-i}
\end{equation}
\vspace{-2mm}
\begin{equation}\label{eq-pi}
	p_i(x'_i) = \int_{X_{-i}}P_i(x'_i,x_{-i})f_{-i}(x_{-i})dx_{-i}
\end{equation}
By equations (\ref{eq-baruid})(\ref{eq-uid})(\ref{eq-hi}) and (\ref{eq-pi}) we can further obtain
\begin{eqnarray}\label{eq-barui2}
	%	\bar{u}_i(x_i,x'_i) = p_i(x'_i) -\sum_{q\in Q}\sum_{a\in\epsilon_i(q)}r^q_{i:a}(x'_i) x_i\\
	\bar{u}_i(x_i,x'_i) = p_i(x'_i) - r_i(x'_i)x_i\label{eq-barui3}
\end{eqnarray}
Thus, the expected utility of agent $i$ when bidding truthfully is:
\begin{equation}\label{eq-Ui}
	\hat{u}_i(x_i) = \bar{u}_i(x_i,x_i)  = p_i(x_i) -r_i(x_i)x_i
\end{equation}
%Implementing a social law by making every agent voluntarily obey the restrictions rather than by ``hard constraints" is one of the challenging issues in social law research. 
Ensuring agents voluntarily comply with the restrictions (rather than via “hard constraints”) is naturally captured by the game-theoretic concept of \emph{Individual Rationality} (IR), which requires that each agent receives a non-negative expected utility after the social law is implemented. For BNIC mechanisms, IR is equivalent to:
%nonnegative expect truthful-bidding  utilities, \emph{i.e.},
\begin{equation}\label{eq-ir}
	\forall i\in A, x_i\in X_i : \hat{u}_i(x_i)\geq 0
\end{equation}

%\vspace*{1mm}
\noindent\textbf{Profit:} For any cost profile  $x\in X$, the selected social law and the payment to agent $i$ will be $\hat{R}(x)$ and $P_i(x)$ respectively.  For a designer with feature set $\mathcal{F}$, the \emph{profit} is
%the value of the structure obtained by implementing the selected social law  is $v_{\mathcal{F}}(S\dag\hat{R}(x))$,the total payment is $\sum_{i\in A}P_i(x)$ and therefore
\begin{equation}
	\sigma(x) = v_{\mathcal{F}}(S\dag\hat{R}(x)) - \sum_{i\in A}P_i(x)
\end{equation}
The expected profit of the designer over the cost profile space $X$ is thus:
\begin{equation}\label{eq-Ep}
	\mathbb{E}_{x\in X}[\sigma(x)]\! =\! \int_X\Big(v_{\mathcal{F}}(S\dag\hat{R}(x))-\sum_{i\in A}P_i(x)\Big)f(x)dx
\end{equation}
The aim of this paper is to find a BNIC and IR mechanism $\langle R^*,P^*\rangle$ that maximize the expected profit.

\subsection{Truthful multiunit profit-optimal mechanisms}
First, we prove that the allocation and payment of any BNIC mechanism can be characterized as follows:

%\begin{lm}\label{lm-U1}
%	A mechanism $\langle H,P\rangle$ is BNIC iff $\forall$ $i\in A$, $x_i, \gamma_i\in X_i:$ $\hat{u}_i(x_i)-\hat{u}_i(\gamma_i)\geq r_i(\gamma_i)(\gamma_i-x_i)$ .
%\end{lm}
\begin{lm}\label{lm-U2}
	A mechanism $\langle R,P\rangle$ is BNIC iff $\forall x_i\in X_i$:
	\begin{itemize}
		\item[1)]   $r_i(x_i)$ is monotone nonincreasing; and
		\item[2)] 	the expected payment to each agent satisfies:
		\vspace{-1mm}
		\begin{equation}
			\label{eq-pixipi0} p_i(x_i) = p_i(0) + x_i r_i(x_i) - \int_{0}^{x_i} r_i(t_i)dt_i
		\end{equation}.
	\end{itemize}	
\end{lm}
\vspace{-2mm}
We denote $\lambda_i(x_i)=x_i + \frac{F_i(x_i)}{f_i(x_i)}$ and refer to it as the \emph{virtual unit cost of agent} $i$. We focus on the \emph{regularity} case of the space $X$, where $\lambda_i$ is non-decreasing function of $x_i$ for every $i$. The above result leads to the following lemmas.

\begin{lm}\label{lm-ep0}
	A mechanism $\langle R,P \rangle$ is BNIC only if 
	%	\vspace{-3mm}
	\begin{equation}\label{eq-Ep2}
		\mathbb{E}_{x\in X}[\sigma(x)]\! =\!\int_X\!\Big(v_{\mathcal{F}}(S\dag\hat{R}(x))- \sum_{i\in A} \lambda_i(x_i)R_i(x)\Big)f(x)dx
	\end{equation}
	%	\vspace{-6mm}
	$$- \sum_{i\in A}\Big(p_i(0)-\int_{0}^{+\infty} r_i(t_i)dt_i\Big)$$
\end{lm}

\begin{lm}\label{lm-g0}
	A mechanism $\langle H,P\rangle$ is BNIC and IR only if	
	\begin{equation}\label{ieq-ph}
		\forall i\in A: p_i(0) - \int_{0}^{+\infty}r_i(t_i)dt_i \geq 0
	\end{equation}
\end{lm}
Moreover, we further have: 
\begin{eqnarray}\label{eq-obj}
	%		g(\eta,x) = v_{\mathcal{F}}(S\dag\eta) - \sum_{i\in A} \big((x_i + \frac{F_i(x_i)}{f_i(x_i)})\cdot\sum_{q\in Q} |\eta_i(q)|\big)\\
	g(\eta,x)	= v_{\mathcal{F}}(S\dag\eta) - \sum_{i\in A} \big(\lambda_i(x_i)\cdot\sum_{q\in Q} |\eta_i(q)|\big)
\end{eqnarray}
where $x\in X$ is the current bid profile and $\eta$ is a selected social law.  Based on the above results,  we present the  $\mathcal{PO}$-$\mathcal{ASL}$ mechanism, as outlined in Algorithm 1. 

\begin{algorithm}[t]
	\caption{$\mathcal{PO}\text{-}\mathcal{ASL}$ Mechanism}
	\label{alg:po-asl}
	\begin{algorithmic}
		%		\STATE \textbf{Mechanism:} $\mathcal{PO}\text{-}\mathcal{ASL}$
		\STATE \textbf{Input:} structure $S = \langle k, Q, q_s, \Pi, \pi, \varepsilon, c, \delta \rangle$, feature set $\mathcal{F}$, p.d.f. $f_i$ of each agent $i$'s cost, and bid profile $x$
		\STATE \textbf{Output:} the allocation and payment $(R^{*}(x), P^{*}(x))$
		
		\STATE 1. Find social law $\eta = \hat{R}^{*}(x) \in \mathcal{SL}_S$ maximizing
		\[
		g(\eta, x) = v_{\mathcal{F}}(S \dagger \eta) - \sum_{i \in A} \left( \lambda_i(x_i) \cdot \sum_{q \in Q} |\eta_i(q)| \right);
		\]
		
		\STATE 2. Compute the payment to each agent $i$ as follows:
		\[
		P_i^{*}(x) = R_i^{*}(x) x_i + \int_{x_i}^{+\infty} R_i^{*}(t_i, x_{-i}) dt_i;
		\]
	\end{algorithmic}
\end{algorithm}

\vspace*{2mm}
\noindent\textbf{Dominant Social Law:} 
An $(i,n)$-social law $\eta^*$ that maximizes $g(\eta,x)$ over all $(i,n)$-social laws under bid profile $x$ is called a dominant $(i,n)$-social law under  $x$. The DOMINANT $(i,n)$-SOCIAL LAW problem refers to, given a structure $S$, a feature set $\mathcal{F}$, a probability density function $f_i$ for each agent $i$’s cost, and a bid profile $x$, finding the dominant $(i,n)$-social law.

The allocation step of the mechanism $\mathcal{PO}$-$\mathcal{ASL}$ is to find the dominant social law under the bid profile $x$. 
%The intuition behind this mechanism is the allocation function $H^*(x)$ maximizes $g$ for every $x$ and fortunately is monotone non-increasing, so it maximizes the integration $\int_X g(\hat{R}(x), x) f(x)dx$;  the payment function satisfies the constraints imposed by BNIC and IR, \emph{i.e.}, equations (\ref{eq-pixipi0}) and (\ref{ieq-ph}), and makes $\sum_{i\in A}\Big(p_i(0)-\int_{0}^{+\infty} r_i(t_i)dt_i\Big)$ get its minimal value $0$, and so the expected profit is maximized within all BNIC and IR mechanisms. 

\begin{lm}\label{lm-qni}
	The allocation function of $\mathcal{PO}$-$\mathcal{ASL}$ is monotone nonincreasing.
\end{lm}

%	The intuition behind the above lemma is actually very obvious. Since as an agent $i$ gradually raises her bid, her virtual cost will increase also, and social laws which restrict more of agent $i$'s actions have an objective function value that decrease faster, and will be possibly replaced by another social law which restricts fewer of agent $i$'s actions. 
Combining these lemmas, we can further prove:

\begin{thm}\label{thm-bio}
	$\mathcal{PO}$-$\mathcal{ASL}$ is BNIC, IR, and maximizes the expected profit among  all BNIC and IR mechanisms.
\end{thm}

It is well-known that for single-parameter procurement settings, truthful mechanisms (where truthful bidding is a dominant strategy) can be characterized as follows:

\begin{thm}~\cite{Archer:FOCS-01,Hartline06}
	A single-parameter procurement auction is truthful if and only if for any agent $i$ and bids of other
	agents $x_{-i}$ fixed, 
	\begin{itemize}
		\item $R_i(x_i,x_{-i})$ is monotone non-increasing.
		\item $P_i(x_i,x_{-i}) = R_i(x_i,x_{-i})x_i+\int_{x_i}^{+\infty}R_i(t_i,x_{-i})dt_i$
	\end{itemize}
\end{thm}

%	Clearly, truthfulness is an even stronger incentive than BNIC that motivates honest bidding. 
We prove that the mechanism $\mathcal{PO}$-$\mathcal{ASL}$ satisfies exactly the above characterizations and is therefore truthful.

\begin{cor}
	$\mathcal{PO}$-$\mathcal{ASL}$ is truthful.
\end{cor}

%It is interesting that we end up discovered a more powerful truthful mechanism, and it maximizes the expected profit among all BNIC and IR mechanisms. %originally aimed to find out a BNIC mechanism but
	
\section{Computation}

The specification of $\mathcal{PO}$-$\mathcal{ASL}$ defines the allocation determination and payment determination problems, but does not provide practical algorithms for solving them. We now address this gap by providing practical algorithms.

\subsection{Reducing payment to allocation} 

In single-unit settings, $P^*_i(x)$ is simply the threshold bid of agent $i$ and can be found by computing the allocation twice, once for  $t_i=0$ and once for $t_i = +\infty$.
However, when each agent may hold multiple units, this shortcut no longer applies. Whether an efficient method exists to compute payments in this setting is an open question.

\vspace*{1mm}
\noindent\textbf{Irregularity in payment calculation in multi-unit settings:}
Computing the payment $P^*_i(x)$ for each agent requires evaluating the integral of $R^*_i(t_i,x_{-i})$, a function that describes how the allocation to agent i changes as its bid $t_i$ varies from $0$ to $+\infty$ (with other agents’ bids $x_{-i}$ fixed). A brute-force approach—computing the allocation for all $t_i\in [0,+\infty]$—is obviously infeasible. We must therefore find a more efficient approach.

%Fortunately, we find out that there is a way to determine  $P^*_i(x)$ based on performing allocation determination several times, and formally prove its correctness.

\vspace*{1mm}
\noindent\textbf{Turning points in the allocation curve, TP:} 
Our first observation is that $R^*_i(t_i,x_{-i})$ must be a monotone nonincreasing \emph{step function} of $t_i$. This follows from Lemma~\ref{lm-qni}, where we proved that $R^*_i(t_i,x_{-i})$ is monotone nonincreasing, and from the fact that $R^*_i(t_i,x_{-i})$ represents the number of actions selected for agent $i$. To determine $R^*_i(t_i,x_{-i})$, it is sufficient to find its turning points, which form the sequence:
$$(p_1, n_1), (p_2,n_2), ..., (p_k, n_k)$$
where $n_1,...,n_k$ are different natural numbers, $(p_j, n_j)$ indicates $t_i = p_j$ is the point where $R^*_i(t_i,x_{-i})$ changes to $n_j$, and therefore $k\in \mathbb{N}\setminus \{0\}$, $\forall i\in [1,k]: p_i\in \mathbb{R}^+$ and $n_k = 0$. Moreover, we also let $p_{0} = x_i$, $n_0 = R^*_i(x_i,x_{-i})$, and $\forall j> k: p_{j} = +\infty$, $n_{j} = 0$, and so we have $\forall j\in [0,k], \forall p \in (p_j, p_{j+1}): R^*_i(p,x_{-i}) = n_j$. 
%The remaining problem is how to determine these turning points efficiently.

%As the bid of agent $i$ increases, the objective function value of a social law decreases if and only if it restricts some of agent $i$'s actions, and always stays the same otherwise. 
%In particular, the decreasing rate of a social law's objective function value is directly proportional to the number of agent $i$'s restricted actions.

\vspace*{1mm}
\noindent\textbf{Deriving an efficient algorithm:}  We begin with the following lemma, which shows that if a social law $\eta$ outperforms another law $\eta'$ (which restricts more of agent $i$’s actions) at bid $x_i$, $\eta$ will continue to outperform $\eta'$ for all higher bids $t\geq x_i$.

\begin{lm}\label{lm-nocatch}
	If $g(\eta,(x_i,x_{-i})) \geq g(\eta',(x_i,x_{-i}))$ and $\sum_{q\in Q} |\eta_i(q)| \leq \sum_{q\in Q} |\eta'_i(q)|$, then $g(\eta,(t,x_{-i})) \geq g(\eta',(t,x_{-i}))$ for all $t\geq x_i$.
\end{lm}
\noindent For any $i\in Ag$ and $n\in\mathbb{N}$, We let
\begin{equation}
	\mathcal{SL}^{(i,n)}_S = \{\eta \in \mathcal{SL}_S~|~ \sum_{q\in Q} |\eta_i(q)| = n\}
\end{equation}
That is, $\mathcal{SL}^{(i,n)}_S$ refers to the set of all possible $(i,n)$-social laws, and moreover, we let 
\begin{equation}\label{eq-dominantSL}
	\eta^{(i,n,(x_i,x_{-i}))} = \arg \max_{\eta\in \mathcal{SL}^{(i,n)}_S} g(\eta, (x_i,x_{-i}))
\end{equation}
\begin{equation}
	\eta^{(x_i,x_{-i})} = \arg \max_{\eta\in \mathcal{SL}_S} g(\eta, (x_i,x_{-i}))
\end{equation}
So, $\eta^{(i,n,(x_i,x_{-i}))}$ and $\eta^{(x_i,x_{-i})}$ refer to \emph{the dominant $(i,n)$-social law} and \emph{the dominant social law} respectively under the bid profile $(x_i,x_{-i})$. 
The following result shows that when agent $i$'s bid varies from $0$ to $+\infty$, the dominant $(i,n)$-social law remains unchanged:

\begin{lm}\label{lm-eta}
	$\forall t_i \in [0,+\infty) : \eta^{(i,n,(t_i,x_{-i}))} = \eta^{(i,n,(0,x_{-i}))}$. 
\end{lm}

Thus, the dominant $(i,n)$-social law is independent of agent $i$’s bid, so we may denote it as $\eta^{(i,n,x_{-i})}$.

\begin{lm}\label{lm-restrict}
	$$\forall t_i \in [x_i,+\infty) : \arg\max_{\eta\in \mathcal{SL}_S} g(\eta, (t_i,x_{-i})) = $$
	$$ \arg\max_{\eta \in \{\eta^{(i,n,x_{-i})}~|~ n\leq R^*_i(x_i,x_{-i}) \}} g(\eta, (t_i,x_{-i})).$$ 
\end{lm}

This result implies that the search space for the optimal social law can be greatly refined. As agent $i$'s bid increases from $x_i$ to $+\infty$, the social law that maximizes the objective function must lie in the set $\{\eta^{(i,R^*_i(x_i,x_{-i}),x_{-i})},..., \eta^{(i,0,x_{-i})}\}$, which contains only $R^*_i(x_i,x_{-i}) + 1$ social laws.

Let$v(n,t)$ denote the objective function value of $\eta^{(i,n,x_{-i})}$  under bid profile $(t,x_{-i})$, and let $v_n = v(n,0)$. This yields the following result for the variation of $v(n,t)$ as agent $i$’s bid varies from $0$ to $+\infty$:
\begin{lm}\label{lm-vnt}
	$v(n,t) = v_n - n\lambda_i(t) $.
\end{lm}

%Now, we are ready to put the pieces together and prove the following theorem for determining the turning points. The underlying idea can be graphically depicted in figure~\ref{fg-pay}. The upper part of this graph shows the curves of $v(5,t),...,v(0,t)$. By lemma~\ref{lm-eta}, each curve depicts a single fixed social law. By lemma~\ref{lm-nocatch}, although all the curves are monotone nonincreasing, $v(m,t)$ declines faster than $v(n,t)$ if $m>n$. When $t=x_i$ the dominant social law  is $\eta^{(i,3,x_{-i})}$, by lemma~\ref{lm-restrict}, as $t$ increases from $x_i$ to $+\infty$, the dominant social law will always be the one in $\{\eta^{(i,3,x_{-i})}, \eta^{(i,2,x_{-i})},\eta^{(i,1,x_{-i})},\eta^{(i,0,x_{-i})}\}$ with the highest objective function value, and graphically it is the one whose curve appears on the ``\emph{top layer}" of this pile of curves, i.e., successively the green curve($A',B$), the yellow curve ($B,C$), the purple curve ($C,D$), and the blue curve ($D,E$), and they correspond exactly to the function $R^*(t_i,x_{-i})$ in the lower part of this graph. The turning points are respectively $B, C$ and $D$. Graphically, theorem~\ref{thm-turning} indicates that the next turning point is just the first intersection point of current top curve with other curves as $t$ increases.

We now combine these results to prove the following theorem for determining the turning points, whose intuition is illustrated in Figure~\ref{fg-pay}. The upper part of the figure shows the curves $v(5,t),...,v(0,t)$. By Lemma~\ref{lm-eta}, each curve corresponds to a fixed social law. By Lemma~\ref{lm-nocatch}, while all curves are monotone non-increasing, $v(n,t)$ declines faster for larger $n$. At $t=x_i$, the dominant social law is $\eta^{(i,3,x_{-i})}$. By Lemma~\ref{lm-restrict}, as t increases from $x_i$ to $+\infty$, the dominant social law will always be in $\{\eta^{(i,3,x_{-i})}, \eta^{(i,2,x_{-i})},\eta^{(i,1,x_{-i})},\eta^{(i,0,x_{-i})}\}$ and will have the highest objective function value. Graphically, this corresponds to the curve that remains on the “\emph{top layer}” of the stack, which successively includes the green curve ($A',B$), yellow curve ($B,C$), purple curve ($C,D$), and blue curve ($D,E$). These curves correspond exactly to the allocation function $R^*(t_i,x_{-i})$  shown in the lower part of the figure, with turning points at $B, C$ and $D$. Theorem~\ref{thm-turning} identifies the first intersection point of the current top curve with any other curve as  $t$ increases:

\begin{figure}[h!]
	\centering
	\includegraphics[width=0.7\linewidth]{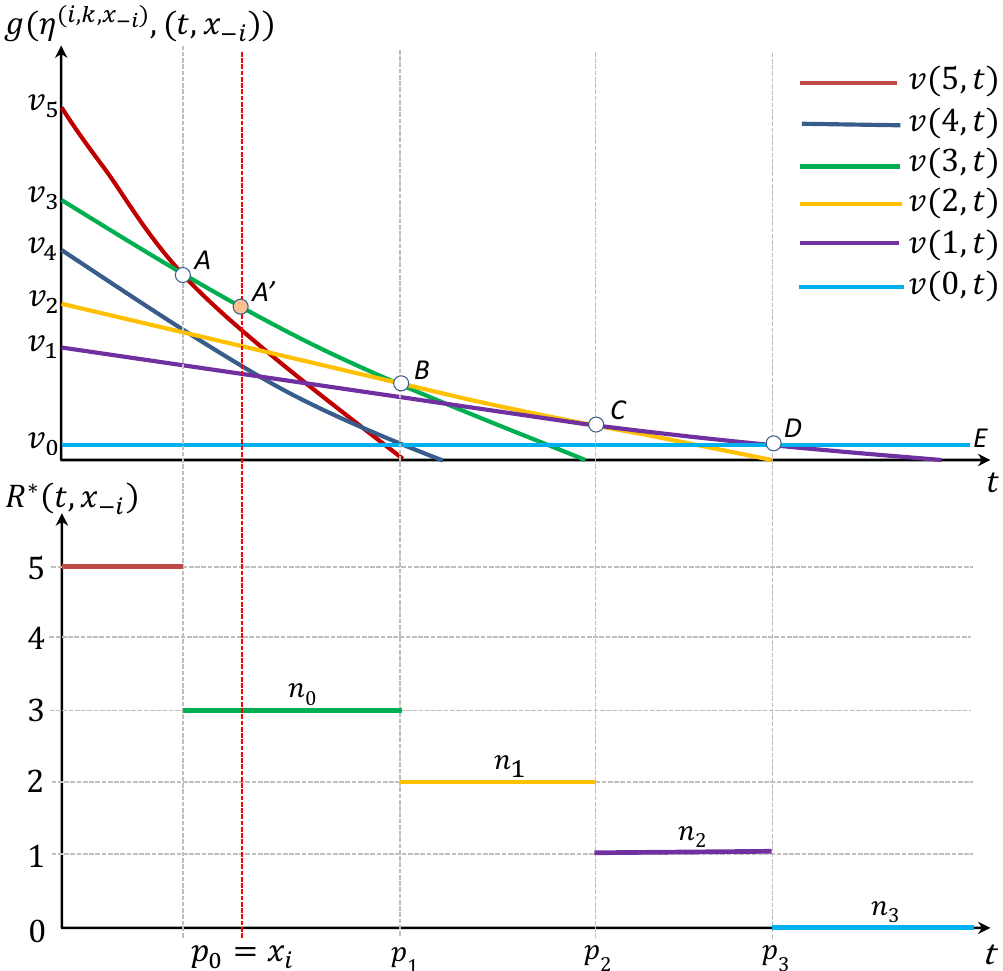}
	\caption{Turning points of agent $i$}
%	\vspace*{-4mm}
	\label{fg-pay}
\end{figure}

\begin{thm}\label{thm-turning}
	$\forall j\geq 1:$ If $n_{j-1} > 0$, then:
	$$p_j = \min_{0 \leq m < n_{j-1}} \lambda^{-1}_i(\frac{v_{n_{j-1}} - v_m}{n_{j-1} -m})$$
	$$n_j = \arg \min_{0 \leq m < n_{j-1}} \lambda^{-1}_i(\frac{v_{n_{j-1}} - v_m}{n_{j-1} -m})$$
	Otherwise, $p_j = +\infty$ and $n_j = 0$.
\end{thm}

With a precise representation of the curve $R^*_i(t_i,x_{-i})$, we can now determine the payments:

\begin{thm}\label{thm-tp}
	Let $k\in \mathbb{N}$ be the smallest index such that $n_k = 0$. Then the payment to agent $i$ is
	
	$$P^*_i(x_i,x_{-i})  = \sum_{1\leq i\leq k} (n_{i-1} - n_i)p_i $$%= n_0 p_1 +\sum_{2\leq i\leq k} n_{i-1}(p_i - p_{i-1})
\end{thm}

We now present algorithm \ref{alg:payment} for payment computation, which uses the subfunctions Allocation($\cdot$) and Allocation-Fix($\cdot$) to find out the dominant social law and the dominant $(i,n)$-social law, respectively. The implementation of these subroutines is deferred to Algorithms \ref{alg:allocation} \& \ref{alg:allocation-fix} later in this paper.

\begin{algorithm}[t]
	\caption{Payment}
	\label{alg:payment}
	\begin{algorithmic}
		\STATE \textbf{Input:} structure $S = \langle k, Q, q_s, \Pi, \pi, \varepsilon, c, \delta \rangle$, feature set $\mathcal{F}$, p.d.f. $f_i$ of each agent $i$'s cost, and bid profile $x$
		
		\FORALL{$i \in Ag$}
		\STATE $\eta^* \leftarrow \text{Allocation}(S, \mathcal{F}, f_1, \dots, f_k, x)$;
		\STATE $n_0 \leftarrow \sum_{q \in Q} |\eta_i^*(q)|$;
		\IF{$n_0 = 0$}
		\STATE $x_i \leftarrow 0$;
		\ELSE
		\FOR{$0 \leq n \leq n_0$}
		\STATE $\eta' \leftarrow \text{Allocation-Fix}(S, \mathcal{F}, f_1, \dots, f_k, x, i, n)$;
		\STATE $v_n \leftarrow g(\eta', (0, x_{-i}))$;
		\ENDFOR \COMMENT{determine all the $(i,n)$-social laws}
		\STATE $n_c \leftarrow n_0$; $j \leftarrow 1$;
		\WHILE{$n_c > 0$}
		\STATE $p_j \leftarrow \min_{0 \leq m < n_{j-1}} \lambda_i^{-1}\left( \frac{v_{n_{j-1}} - v_m}{n_{j-1} - m} \right)$;
		\STATE $n_j \leftarrow \arg\min_{0 \leq m < n_{j-1}} \lambda_i^{-1}\left( \frac{v_{n_{j-1}} - v_m}{n_{j-1} - m} \right)$;
		\STATE $n_c \leftarrow n_j$; $j++$;
		\ENDWHILE \COMMENT{find the TPs by theorem 14}
		\STATE $x_i \leftarrow \sum_{1 \leq i \leq j-1} (n_{i-1} - n_i) p_i$;
		\STATE \COMMENT{compute agent $i$'s payment by theorem 15}
		\ENDIF
		\ENDFOR
		\RETURN $(x_1, \dots, x_k)$
	\end{algorithmic}
\end{algorithm}
%	\vspace{2mm}

%	The correctness of algorithm 1 and its underlying idea can be captured by the following corollary and its proof.

\begin{cor}
	Algorithm 1 correctly computes the payment of mechanism $\mathcal{PO}$-$\mathcal{ASL}$.
\end{cor}

\subsection{Computational complexity} Allocation determination is the basic building block of the proposed mechanism, however we can show it is intractable.

\begin{lm}\label{lm-cdsl}
	Both {\sc Dominant Social Law} and {\sc Dominant $(i,n)$-Social Law} are FP$^{NP}$-complete.
\end{lm}

\begin{thm}
	Mechanism $\mathcal{PO}$-$\mathcal{ASL}$ is FP$^{NP}$-complete.
\end{thm}

Thus, we are unlikely to find efficient algorithms for directly computing $\mathcal{PO}$-$\mathcal{ASL}$. Methodologies for countering intractability should therefore be considered.

\vspace*{-1mm}	
\subsection{An ILP-based algorithm} 
%Since Integer-linear programming (ILP) is one of the most successful and widely-used approaches to solving computationally hard optimization problems, we are going to try to find out a solution based on ILP, and so it can be handled by current ILP solvers, which have already extensively used in solving industrial level large scale ILPs. 
We propose an approach that, given an instance of {\sc Dominant ($(i,n)$-)Social Law}, automatically constructs an ILP such that solutions to the ILP correspond exactly to solutions of the original problem instance. %The idea is the value of the objection function $g(\eta,x)$, which we aim to maximize, depends not only on whether each transition is selected but also on whether each formula $\varphi_i$ in the feature set $\mathcal{F} = \{(\varphi_1,c_1),...,(\varphi_k,c_k)\}$ is satisfied in the current state, which may further depend on whether each sub-formula of $\varphi_i$ is satisfied in the current state $s$ or the next states. Each of the above choice can be captured by a boolean variable, and the combined value assignments to all the boolean variables are constrained by the ATL semantics.

\vspace{1mm}	
\noindent\textbf{Closure:} Let $cl(\varphi)$ denote the \emph{closure} of a formula $\varphi$, i.e.,
\vspace*{-1mm}
\begin{equation}\label{eq-cl1}
	cl(\varphi) = \{\varphi\}\cup sub(\varphi)
\end{equation}
\vspace*{-2mm}
where
\begin{equation}\label{eq-sub}
	sub(\varphi)\!=\!
	\begin{cases}
		cl(\psi)\!\cup\! cl(\chi), &\!\! \text{if}~\varphi = \psi\vee\chi~\text{or}~\llangle A\rrangle\psi\mathcal{U}\chi \\
		cl(\psi), &\!\! \text{if}~\varphi = \neg \psi~\text{or}~\llangle A\rrangle \bigcirc \psi~\text{or}~\llangle A\rrangle\!\Box\! \psi\\
		\{\varphi\}, &\!\! \text{if}~\varphi\in \Pi
	\end{cases}
\end{equation}
	According to the above definition, the closure of a formula intrinsically forms a tree-like structure.
E.g., Example 2 (in the appendix). The number of formulas in $cl(\varphi)$ is called the \emph{size} (or \emph{length}) of formula $\varphi$. A formula's size directly reflects its construction complexity and closely connects to the complexity of related algorithms. For example, ATL model checking has a $\mathcal{O}(t\cdot l)$ time algorithm, where $t$ is the number of transitions in the given concurrent game structure and $l$ is the size of the given formula~\cite{Alur:JACM-02}. We also use these parameters to characterize the size of instances of social law synthesizing problem.

For any feature set $\mathcal{F}=\{(\varphi_1,c_1),...,(\varphi_n,c_n)\}$, we let
\begin{equation}\label{eq-cl}
	cl(\mathcal{F})=cl(\varphi_1)\cup...\cup cl(\varphi_n)
\end{equation}
Actually, $cl(\mathcal{F})$ is the set of all the formulas which can possibly influence the valuation of the objective function. 

\vspace{1mm}	
\noindent\textbf{Variables:}
We can first of all introduce the following 4 classes of boolean variables:
\begin{itemize}
	\item[1)] $x^q_\varphi\in\{0,1\},\forall q\in Q, \varphi\in cl(\mathcal{F})$
	\item[2)] $y^q_{i:a}\in\{0,1\}, \forall q\in Q, i\in Ag, a\in \varepsilon_i(q)$
	\item[3)] $y^q_{A:\vec{m}_A}\in\{0,1\},\forall q\in Q, A\subseteq Ag, \vec{m}_A\in D_A(q)$
	\item[4)] $z^{q,\varphi}_{A:\vec{m}_A}\!\!\!\in\!\{0,1\},\forall q\in Q,\varphi\in cl(\mathcal{F}),A\subseteq Ag,\vec{m}_A\in D_A(q)$
\end{itemize}
where
\begin{itemize}
	\item[1)] $x^q_\varphi=1$ iff $S\dag\eta,q\vDash\varphi$;
	\item[2)] $y^q_{i:a}=1$ iff at state $q$, agent $i$'s action $a$ is forbidden;
	\item[3)] $y^q_{A:\vec{m}_A}=1$ iff at state $q$, joint action $\vec{m}_A$ is forbidden;
	\item[4)] $z^{q,\varphi}_{A:\vec{m}_A}=1$ iff at state $q$, coalition $A$ adopting the joint action $\vec{m}_A$ can guarantee that $\varphi$ is satisfied in the next state.
\end{itemize}

We now introduce additional intermediate variables and transform these constraints into an equivalent set of ILP constraints, yielding the following ILP for solving the {\sc Dominant Social Law} problem. 

Note that, we focus on the state set $Q$, formula set $cl(\mathcal{F})$, and agent set $Ag$, so for example we write $\forall q,i,\varphi$ as an abbreviation for $\forall q\in Q,i\in Ag,\varphi\in cl(\mathcal{F})$. When the state $q$ is clear from the context, we also write $\forall \vec{m}_A$ as an abbreviation for $\forall \vec{m}_A\in D_A(q)$, and write $\forall \vec{m}_{\bar{A}}$ as an abbreviation for $\forall \vec{m}_{\bar{A}}\in D_{\bar{A}}(q)$. Moreover, $\forall \llangle A\rrangle$ means for every operator $\llangle A\rrangle$ that appear in a formula in $cl(\mathcal{F})$. %
%\vspace{2mm}
\begin{center}
\noindent\fbox{%  
	\parbox{445pt}{%  
		\begin{flushleft}  
			{\sc ILP-Dom-SL}($S, \mathcal{F}, f_1,...,f_k, x$): %{\sc ILP-Dominant-IN-SL}($S, \mathcal{F}, f_1,...,f_k, x, i, n$)
			maximize 
			\vspace{-3mm}
			\begin{equation}\label{ilp-object}
				\sum_{(\varphi_i,c_i)\in\mathcal{F}}c_j\cdot x^{q_s}_{\varphi_j}-\sum_{i\in Ag}\sum_{q\in Q}\sum_{a\in \epsilon_i(q)}(x_i + \frac{F_i(x_i)}{f_i(x_i)}) y^q_{i:a}
			\end{equation}
			%				\qquad\\
			subject to: 
			\vspace{-4mm}
			\setlength{\jot}{1pt} 
			\begin{eqnarray}
				x^q_{\varphi}\in \{0,1\},~~\forall q, \varphi\label{ilp-c1}\\                 
				y^q_{i:a}\in\{0,1\},~~\forall q, i, a\in \varepsilon_i(q) \label{ilp-c2}\\
				\sum_{a\in\varepsilon_i(q)}(1-y^q_{i:a})\geq 1, \forall q,i \label{ilp-c14}\\
				y^q_{A:\vec{m}_A}\in\{0,1\},~~\forall q, \llangle A\rrangle~\text{or}~\llangle \bar{A}\rrangle, \vec{m}_A \label{ilp-c3}\\
				y^q_{A:\vec{m}_A} \geq y^q_{i:\vec{m}_A[i]}, ~~\forall q, \llangle A\rrangle~\text{or}~\llangle \bar{A}\rrangle, \vec{m}_A, i \label{ilp-c5}\\
				y^q_{A:\vec{m}_A} \leq \sum_{i\in A} y^q_{i:\vec{m}_A[i]}, ~~\forall q, \llangle A\rrangle~\text{or}~\llangle \bar{A}\rrangle, \vec{m}_A\label{ilp-c6}\\
				s^{q,\varphi}_{A:\vec{m}_A, \vec{m}_{\bar{A}}}\in \{0,1\},
				\forall q, \varphi, \llangle A\rrangle,\vec{m}_A,\vec{m}_{\bar{A}}\label{ilp-c7}\\
				s^{q,\varphi}_{A:\vec{m}_A, \vec{m}_{\bar{A}}} \geq y^q_{\bar{A}:\vec{m}_{\bar{A}}},
				\forall q, \varphi, \llangle A\rrangle,\vec{m}_A,\vec{m}_{\bar{A}}\label{ilp-c8}\\
				s^{q,\varphi}_{A:\vec{m}_A, \vec{m}_{\bar{A}}} \geq x^{\delta(q, (\vec{m}_A,\vec{m}_{\bar{A}}))}_\varphi, \forall q, \varphi, \llangle A\rrangle,\vec{m}_{A}, \vec{m}_{\bar{A}}\label{ilp-c9}\\
				s^{q,\varphi}_{A:\vec{m}_A, \vec{m}_{\bar{A}}}\! \leq\! y^q_{\bar{A}:\vec{m}_{\bar{A}}}\!\!\!\!+\! x^{\delta(\!q\!, (\!\vec{m}_A,\!\vec{m}_{\bar{A}}\!)\!)}_\varphi, \forall q,\! \varphi, \!\llangle A\rrangle,\!\vec{m}_A,\!\vec{m}_{\bar{A}}\label{ilp-c10}\\ 
				z^{q,\varphi}_{A:\vec{m}_A}\in\{0,1\},
				\forall q, \varphi, \llangle A\rrangle, \vec{m}_A\label{ilp-c4}\\						z^{q,\varphi}_{A:\vec{m}_A}\leq s^{q,\varphi}_{A:\vec{m}_A, \vec{m}_{\bar{A}}},\forall q, \varphi, \llangle A\rrangle,\vec{m}_A,\vec{m}_{\bar{A}}\label{ilp-c11}\\ 
				z^{q,\varphi}_{A:\vec{m}_A}\geq\! 1\!-\!\!\!\!\sum_{\vec{m}_{\bar{A}}\in D_{\bar{A}}(q)} \!(1-s^{q,\varphi}_{A:\vec{m}_A, \vec{m}_{\bar{A}}}),\forall q, \varphi, \llangle A\rrangle,\vec{m}_A\label{ilp-c12}\\ 
				x^q_p=1, \forall q, p\in (\Pi\cap cl(\mathcal{F}))\cap \pi(q)\label{ilp-c15}\\
				x^q_p=0, \forall q, p\in (\Pi\cap cl(\mathcal{F}))\setminus \pi(q)\label{ilp-c16}\\
				x^q_{\neg\varphi} =  1- x^q_{\varphi}, ~~\forall q, \neg\varphi\label{ilp-c17}\\
				x^q_{\psi\vee\chi} \geq  x^q_{\psi},~~\forall q, \psi\vee\chi\label{ilp-c18}\\	
				x^q_{\psi\vee\chi} \geq  x^q_{\chi},~~\forall q, \psi\vee\chi\label{ilp-c19}\\
				x^q_{\psi\vee\chi} \leq  x^q_{\psi}+x^q_{\chi},~~\forall q, \psi\vee\chi\label{ilp-c20}\\
				e^{q,\varphi}_{A:\vec{m}_A}\in \{0,1\}, \forall q, \varphi, \llangle A\rrangle, \vec{m}_A\label{ilp-c21}\\
				e^{q,\varphi}_{A:\vec{m}_A}\geq z^{q,\varphi}_{A:\vec{m}_A}-y^q_{A:\vec{m}_A},\forall q, \varphi, \llangle A\rrangle, \vec{m}_A\in q\label{ilp-c22}\\ 
				e^{q,\varphi}_{A:\vec{m}_A}\leq z^{q,\varphi}_{A:\vec{m}_A}, \forall q, \varphi, \llangle A\rrangle, \vec{m}_A\label{ilp-c23}\\
				e^{q,\varphi}_{A:\vec{m}_A}\leq 1- y^q_{A:\vec{m}_A}, \forall q, \varphi, \llangle A\rrangle, \vec{m}_A\label{ilp-c24}\\
				x^q_{\llangle A\rrangle \bigcirc\varphi} \geq e^{q,\varphi}_{A:\vec{m}_A},~~\forall q, \llangle A\rrangle\bigcirc\varphi,\vec{m}_A\label{ilp-c25}\\	
				x^q_{\llangle A\rrangle \bigcirc\varphi} \leq \sum_{\vec{m}_A\in D_A(q)} e^{q,\varphi}_{A:\vec{m}_A},~~\forall q, \llangle A\rrangle\bigcirc\varphi\label{ilp-c26}\\
				r^q_{\llangle A\rrangle \psi\mathcal{U}\chi} \in \{0,1\},~~\forall q, \llangle A\rrangle\psi\mathcal{U}\chi\label{ilp-c27}\\
				r^q_{\llangle A\rrangle \psi\mathcal{U}\chi} \leq x^q_{\psi},~~\forall q, \llangle A\rrangle\psi\mathcal{U}\chi\label{ilp-c28}\\
				r^q_{\llangle A\rrangle \psi\mathcal{U}\chi} \leq \sum_{\vec{m}_A\in D_A(q)} e^{q,\llangle A\rrangle\psi\mathcal{U}\chi}_{A:\vec{m}_A}, \forall q, \llangle A\rrangle\psi\mathcal{U}\chi\label{ilp-c29}\\
				r^q_{\llangle A\rrangle \psi\mathcal{U}\chi} \geq  x^q_{\psi}+ e^{q,\!\llangle\! A\!\rrangle\psi\mathcal{U}\chi}_{A:\vec{m}_A}-1,\! \forall q,\! \llangle\! A\!\rrangle\psi\mathcal{U}\chi,\!\vec{m}_A\label{ilp-c30}\\
				x^q_{\llangle A\rrangle \psi\mathcal{U}\chi} \geq x^q_{\chi},\forall q, \llangle A\rrangle\psi\mathcal{U}\chi\label{ilp-c31}\\
				x^q_{\llangle A\rrangle \psi\mathcal{U}\chi} \geq r^q_{\llangle A\rrangle \psi\mathcal{U}\chi},\forall q, \llangle A\rrangle\psi\mathcal{U}\chi\label{ilp-c32}\\
				x^q_{\llangle A\rrangle \psi\mathcal{U}\chi} \leq x^q_{\chi}+r^q_{\llangle A\rrangle \psi\mathcal{U}\chi},\forall q, \llangle A\rrangle\psi\mathcal{U}\chi\label{ilp-c33}\\
				x^q_{\llangle A\rrangle \Box\varphi} \leq x^q_{\varphi},\forall q, \llangle A\rrangle\!\Box\!\varphi\label{ilp-c34}\\
				x^q_{\llangle A\rrangle \Box\varphi} \leq \sum_{\vec{m}_A\in D_A(q)} e^{q,\llangle A\rrangle\Box\varphi}_{A:\vec{m}_A},\forall q, \llangle A\rrangle\!\Box\!\varphi\label{ilp-c35}\\   	
				x^q_{\llangle A\rrangle \Box\varphi}\! \geq   x^q_{\varphi}+ e^{q,\llangle\! A\!\rrangle\Box\varphi}_{A:\vec{m}_A} - 1,\forall q, \llangle\! A\!\rrangle\!\Box\!\varphi,\vec{m}_A\label{ilp-c36} 
			\end{eqnarray}		
		\end{flushleft}  
	} 
} 
\end{center}

We define {\sc ILP-Dom-IN-SL}($S, \mathcal{F}, f_1,...,f_k, x, i, n$) as the ILP obtained from {\sc ILP-Dom-SL}($S, \mathcal{F}, f_1,...,f_k, x$) by aiding the additional constraint:
\begin{equation}\label{ilp-c4'}
	\sum_{q\in Q}\sum_{a\in \epsilon_i(q)} y^q_{i:a} = n 
\end{equation}
%$\eta_{\ell^i_{s_{m-1}}}$
The dependencies between the variables can be vividly depicted as figure~\ref{fig:dep}.% that is, the assignments to variables of the form $x^q_{\llangle A\rrangle \bigcirc\varphi}$, $x^q_{\llangle A\rrangle \varphi\mathcal{U}\psi}$ and $x^q_{\llangle A\rrangle \Box\varphi}$ finally depend on the assignments to the variables of the form $x^q_\varphi$ and $y^q_{i:a}$, through some intermediate variables.
%\vspace{-4mm}
\begin{figure}[h]
	\centering
	\includegraphics[width=0.55\linewidth]{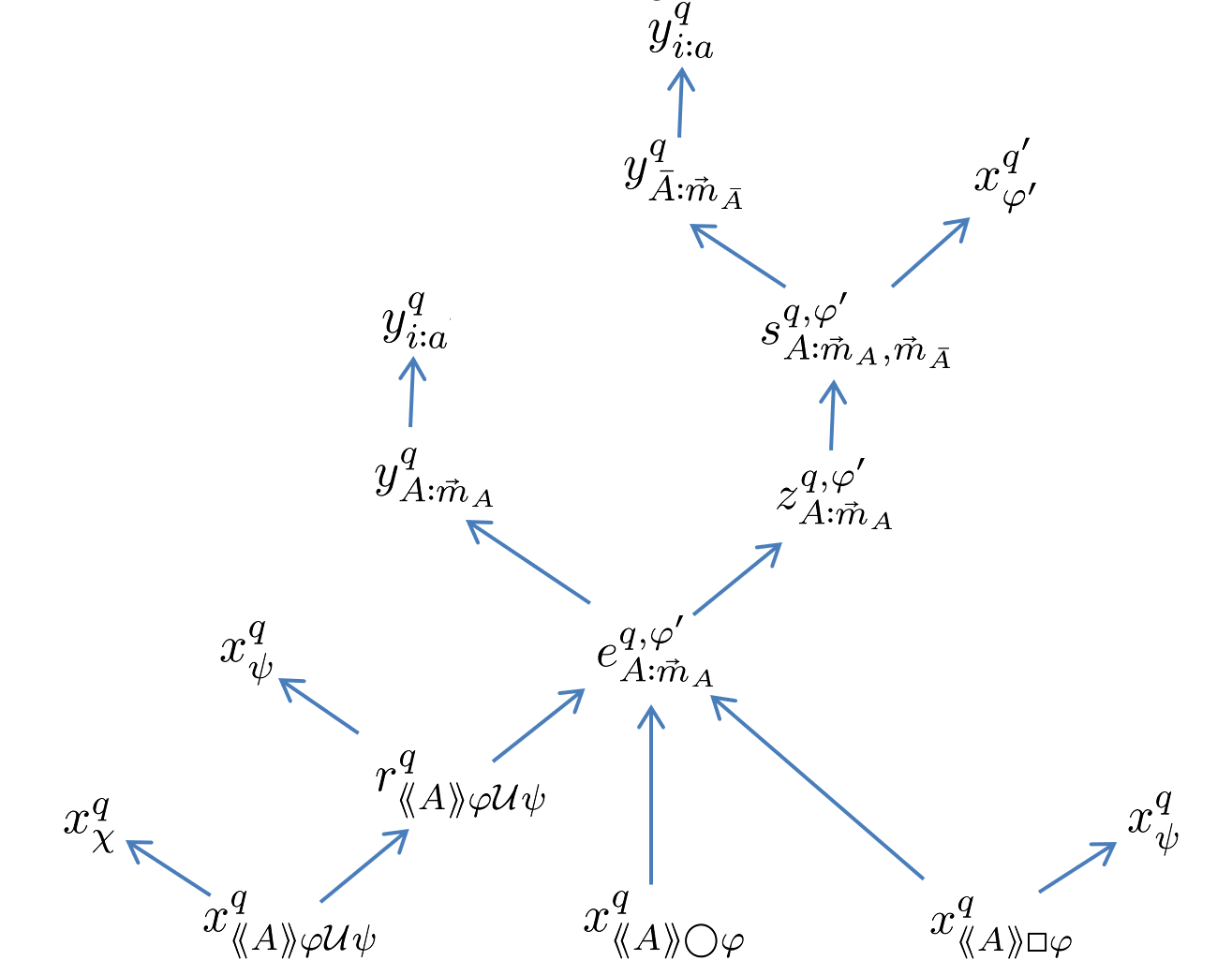}
	\caption{The dependencies between the variables}
%	\vspace*{-3mm}
	\label{fig:dep}
\end{figure}

The solution to {\sc ILP-Dom-SL} or {\sc ILP-Dom-IN-SL} is actually an assignment $\ell$, which assigns an $0$ or $1$ to each of the variables under all the constraints. We let $\eta_\ell(i,q)$ be the social law corresponding to assignment $\ell$, that is,
\begin{equation}
	\eta_\ell(i,q)=\{a\in \varepsilon_i(q)| \ell(y^q_{i:a})=1\}
\end{equation}
Then, we can further prove the following results:

\begin{lm}\label{lm-x16}~		
	\begin{itemize}
		\item[1)] $\forall q, i, a$: $\ell(y^q_{i:a})=1$ iff $a\in\eta_\ell(i,q)$;
		\item[2)] $\forall q, A, \vec{m}_A$: $\ell(y^q_{A:\vec{m}_A})=1$ iff $\exists i\in A: \ell(y^q_{i:\vec{m}_A[i]}) = 1$;
		\item[3)] $\forall q, \llangle A\rrangle, \vec{m}_A,\vec{m}_{\bar{A}}$: $\ell(s^{q,\varphi}_{A:\vec{m}_A, \vec{m}_{\bar{A}}}) = 1$ iff $\ell(y^q_{\bar{A}:\vec{m}_{\bar{A}}})=1$ or $\ell(x^{\delta(q, (\vec{m}_A,\vec{m}_{\bar{A}}))}_\varphi) = 1$;
		\item[4)] $\ell(z^{q,\varphi}_{A:\vec{m}_A})=1$ iff $\forall \vec{m}_{\bar{A}}: $$\ell(s^{q,\varphi}_{A:\vec{m}_A, \vec{m}_{\bar{A}}}) = 1$;
		\item[5)] $\forall q, \varphi, A, \vec{m}_A$: $\ell(e^{q,\varphi}_{A:\vec{m}_A}) = 1$ iff $\ell(y^q_{A:\vec{m}_A})=0$ and ;
		\item[6)] $\forall q, \llangle A\rrangle\psi\mathcal{U}\chi$: $\ell(r^q_{\llangle A\rrangle \psi\mathcal{U}\chi}) = 1$ iff $\ell(x^q_{\psi})=1$ and $\exists \vec{m}_A: (\ell(y^q_{A:\vec{m}_A})=0$ and $\ell(z^{q,\llangle A\rrangle\psi\mathcal{U}\chi}_{A:\vec{m}_A})=1)$.
	\end{itemize}
\end{lm}

\begin{lm}\label{lm-c} For any $\varphi\in cl(\mathcal{F})$, $\ell(x^q_\varphi)=1$ iff $S\dag\eta_{\ell},q\vDash\varphi$.
\end{lm}

Finally, the following theorem follows from the above two lemmas. It shows that the proposed ILP correctly computes the allocation function.

\begin{thm}\label{thm-17}~~
	\begin{itemize}
		\item[1)] If $\ell$ is a solution to {\sc ILP-Dom-SL}, then $\eta_\ell$ is a dominant social law under the bid profile $x$;
		\item[2)] If $\ell$ is a solution to {\sc ILP-Dom-IN-SL}, then $\eta_\ell$ is a dominant $(i,n)$-social law under the bid profile $x$.
	\end{itemize}	
\end{thm}

% Algorithm 2: Allocation
\begin{algorithm}[t]
	\caption{Allocation}
	\label{alg:allocation}
	\begin{algorithmic}
		\STATE \textbf{Input:} structure $S = \langle k, Q, q_s, \Pi, \pi, \varepsilon, c, \delta \rangle$, feature set $\mathcal{F}$, p. d. f. $f_i$ of each agent $i$'s cost, and bid profile $x$
		
		\STATE Generate {\sc ILP-Dom-SL}$(S, \mathcal{F}, f_1, \dots, f_k, x)$;
		\STATE $\ell \leftarrow $ Solving  {\sc ILP-Dom-SL} $(S, \mathcal{F}, f_1, \dots, f_k, x)$;
		\FORALL{$i \in Ag$, $q \in Q$}
		\STATE $\eta_\ell(i,q) \leftarrow \{a \in \varepsilon_i(q) \mid \ell(y_{i:a}^q) = 1\}$
		\ENDFOR
		\RETURN $\eta_\ell$
	\end{algorithmic}
\end{algorithm}
%\vspace{-2mm}	
% Algorithm 3: Allocation-Fix
\begin{algorithm}[t]
	\caption{Allocation-Fix}
	\label{alg:allocation-fix}
	\begin{algorithmic}
		\STATE \textbf{Input:} structure $S = \langle k, Q, q_s, \Pi, \pi, \varepsilon, c, \delta \rangle$, feature set $\mathcal{F}$, p. d. f. $f_i$ of each agent $i$'s cost, bid profile $x$, $i \in Ag$ and $n \in \mathbb{N}$
		
		\STATE Generate  {\sc ILP-Dom-IN-SL} $(S, \mathcal{F}, f_1, \dots, f_k, x, i, n)$;
		\STATE $\ell \leftarrow $ Solving {\sc ILP-Dom-IN-SL}$(S, \mathcal{F}, f_1, \dots, f_k, x, i, n)$;
		\FORALL{$i \in Ag$, $q \in Q$}
		\STATE $\eta_\ell(i,q) \leftarrow \{a \in \varepsilon_i(q) \mid \ell(y_{i:a}^q) = 1\}$
		\ENDFOR
		\RETURN $\eta_\ell$
	\end{algorithmic}
\end{algorithm}

\begin{thm}\label{thm-cipl}
	Both {\sc ILP-Dom-SL} and {\sc ILP-Dom-IN-SL} can be generated in $\mathcal{O}(|Q|\cdot t\cdot l^2)$ time, where $|Q|$, $t$ are respectively the state number and state transition number of the given structure, and $l$ is the total length of the formulas in the given feature set.
\end{thm}

%	Algorithm 2 and algorithm 3 generate the corresponding ILPs in polynomial-time (according to theorem~\ref{thm-cipl}) and can be efficiently handled based on current ILP solvers, and according to theorem~\ref{thm-17} they correctly find out the dominant social law and dominant $(i,n)$-social law. Moreover, we have already shown that algorithm 1 correctly determines the payment to each agent via calling the above 2 algorithms. 
Therefore, via designing algorithm 2$\sim$4, we have proposed a practical way to compute the proposed mechanism. 

%\section{Related Work}

%The cost of implementing a social law was first studied by~\cite{Fitoussi:AIJ-00}. ~\cite{Agotnes:AAMAS-10} later proposed the OPTIMAL SOCIAL LAW problem, which balances the cost and benefit of implementing a social law and models CTL-based social law synthesizing as a combinatorial optimization problem.  \cite{Wu:AAMAS-17a,Wu:JAAMAS-21} non-trivially extended the {\sc Optimal Social Law} problem to the strategic case, developed solutions based on the framework of optimal auctions~\cite{Myerson:MOR-81,Elkind:SODA-04}. We consider more general ATL-based social laws, which leads to a new \emph{multiunit auction} design problem~\cite{Maskin:EMIG-89,Bhattacharya:TCS-20,Simina:AIJ-23,chan2014truthful,Wu:AAMAS-19,Wu:AAAI-20,Hagen2023Collusion,Gautier2023Multi,Potfer2024Improved,Bougt2025Revenue,Aberg2025Quantifying}  where payment determination is a major challenge. %A more comprehensive review is provided in the appendix.

\section{Conclusion}
In this paper, we have systematically studied Social Law Synthesis (SLS) in strategic multi-agent environments as a novel multi-unit profit-optimal mechanism design problem. By modeling SLS within a Bayesian single-parameter procurement auction framework built on Alternating-time Temporal Logic (ATL), we have addressed key challenges arising from agents’ private cost information, strategic behavior, and the multi-unit nature of action restrictions.
Our main theoretical contributions include a representation lemma that compactly encodes valid valuations via ATL feature sets, a polynomial-time reduction from payment computation to allocation determination, and a complexity result showing that allocation is \(FP^{NP}\)-complete. To achieve practical computability, we encode the full ATL semantics into integer linear programming constraints, enabling efficient solution using off-the-shelf ILP solvers. The resulting PO-ASL mechanism is proven to be truthful, individually rational, and expected-profit-maximizing among all incentive-compatible mechanisms.
This work bridges formal methods for multi-agent coordination and algorithmic mechanism design, offering a principled and implementable approach to norm synthesis under self-interested agents. Future directions will extend the model to general cost structures, online and dynamic settings, coalition formation, and approximate mechanisms for large-scale systems.
	% 参考文献（arxiv 直接用 plain 即可）
\bibliographystyle{plain}
\bibliography{myref}

\appendix
\section{Examples for the Problem Setting}

The following example shows that coordinating a system consisting of some rational agents, can be modeled as a Social Law Synthesizing (SLS) problem. 

\begin{ex}
	Consider a system consists of two agents and a shared storage. Each agent also has an independent cache of their own. During the running process, each agent constantly generates data and writes them to their own cache. But the space of each agent's cache is limited, they have to transfer the data in it to the shared storage when it becomes full. To ensure correctness, each agent's access to the shared storage should be mutually exclusive.
	
	\begin{table*}[h] 
		\caption{The feature set $\mathcal{F}$}
		\begin{tabular}{m{4.5cm} m{0.8cm} m{11.5cm}}
			\hline
			\hline
			Property & Value & Description\\
			\hline
			$\varphi_1 = \llangle \rrangle \Box \neg \epsilon$ & 30  &The system will never enter the error state.\\
			
			$\varphi_2 = \llangle \rrangle \Box (\alpha_1\rightarrow \llangle1\rrangle \bigcirc \beta_1)$ & 10  &In agent 1 priority state, agent 1 can control the storage immediately.\\
			
			$\varphi_3 = \llangle \rrangle \Box (\alpha_2\rightarrow \llangle2\rrangle \bigcirc \beta_2)$ & 10 &In agent 2 priority state, agent 2 can control the storage immediately.\\
			
			$\varphi_4 = \llangle \rrangle \Box (\beta_1\rightarrow \llangle1\rrangle \bigcirc \alpha_2)$ & 10 &In agent 1 control state, agent 1 can write to the storage immediately.\\
			
			$\varphi_5 = \llangle \rrangle \Box (\beta_2\rightarrow \llangle2\rrangle \bigcirc \alpha_1)$ & 10 &In agent 2 control state, agent 2 can write to the storage immediately.\\
			
			$\varphi_6 = \llangle \rrangle \Box \llangle 1\rrangle \Diamond \beta_1$ & 12 &Agent 1 can always try to eventually control the storage.\\
			
			$\varphi_7 = \llangle \rrangle \Box \llangle 2\rrangle \Diamond \beta_2$ & 12 &Agent 2 can always try to eventually control the storage.\\
			
			$\varphi_8 = \llangle \rrangle \Box (\alpha_1\rightarrow \llangle1,2\rrangle \Diamond \beta_2)$ & 6 &In agent 1 priority state, it is possible for agent 2 to eventually control the storage.\\
			
			$\varphi_9 = \llangle \rrangle \Box (\alpha_2\rightarrow \llangle1,2\rrangle \Diamond \beta_1)$ & 6 &In agent 2 priority state, it is possible for agent 1 to eventually control the storage.\\
			
			$\varphi_{10} = \llangle \rrangle \Box (\beta_1\rightarrow \llangle 1\rrangle\bigcirc \beta_1)$ & 4 &In agent 1 control state, agent 1 can keep controlling the storage.\\
			
			$\varphi_{11} = \llangle \rrangle \Box (\beta_2\rightarrow \llangle 2\rrangle\bigcirc \beta_2)$ & 4 &In agent 2 control state, agent 2 can keep controlling the storage.\\
			\hline
			%\hline
		\end{tabular}
		\label{tb-fs}
	\end{table*}
	
	\begin{table*}[h] 
		\caption{Social laws and their profits (Each table line specifies a social law $\eta_i$ (where for example $0^1_a$ means agent 1's action $a$ in state $1$ is forbidden by $\eta_i$), labels the satisfied formulas (with ``+") and unsatisfied formulas (with ``-") from $\mathcal{F}$ (in the structure obtained by implementing $\eta_i$), and calculates the cost ($\sum_{i\in Ag} (c(i)\cdot\sum_{q\in Q} |\eta_i(q)|)$), value ($v_{\mathcal{F}}(S\dag \eta_i)$) and profit ($g(\eta_i)$).}
		%		\resizebox{\textwidth}{!}{
			\begin{tabular}{c c c c c c c c c c c c c c c c}
				\hline
				\hline
				& Social law  & $\varphi_1$ & $\varphi_2$ & $\varphi_3$ & $\varphi_4$ & $\varphi_5$ & $\varphi_6$ & $\varphi_7$ & $\varphi_8$ & $\varphi_9$  & $\varphi_{10}$ & $\varphi_{11}$ & Cost & Value & Profit\\
				\hline
				$\eta_0$ &  $\emptyset$ & - & + & + & - & - & - & - & + & + & - & - & 0 & 32 & 32\\
				$\eta_1$ &  $\{0^1_a, 0^2_a\}$ & + & - & + & + & + & - & - & - & + & + & + & 25 & 74 & 49\\
				$\eta_2$ & $\{0^1_a, 1^1_a, 3^1_w\}$ & + & - & + & + & + & - & + & + & - & + & + & 30 & 86 & 56\\
				$\eta_3$ & $\{2^2_w, 3^1_w\}$ & + & + & + & + & + & - & - & + & + & + & +  & 25 & 90 & 65\\
				$\eta_4$ & $\{0^2_a, 1^2_a, 2^2_w\}$ & + & + & - & + & + & + & - & - & + & + & + & 45 & 86 & 41\\
				$\eta_5$ & $\{2^2_w, 3^1_w,3^2_r\}$ & + & + & + & + & + & + & - & + & + & + & - & 40 & 98 & 58\\
				$\eta_6$ & $\{2^1_w, 3^2_w\}$ & + & + & + & - & - & - & - & + & + & - & - & 25 & 62 & 37\\
				$\eta_7$ & $\{2^2_w, 3^1_w,2^1_r,3^2_r\}$ & + & + & + & + & + & + & + & + & + & - & - & 50 & 106 & 56\\
				\hline
				\hline
			\end{tabular}%}
		\label{tb-slu}
	\end{table*}
	
	The running of such a systems can be represented as a {\sc ccgs} $S=\langle k, Q, q_s,\Pi, \pi, \varepsilon, c, \delta\rangle $, where
	\begin{itemize}
		\item $k = 2$, i.e., the agent set is $Ag= \{1,2\}$;
		\item $Q  = \{q_0, q_1, q_2, q_3, q_4\}$, $q_s = q_0$, the initial state of the system is $q_0$;
		
		\item $\Pi = \{\alpha_1, \alpha_2, \beta_1, \beta_2, \epsilon\}$, and captures the following meanings:\\
		$\alpha_i:$ Agent $i$ has the priority to apply for controlling the shared storage;\\
		$\beta_i:$ Agent $i$ has been authorized to control the shared storage;\\
		$\epsilon:$ Writing to the storage has violated mutual exclusiveness.
		\item Label function $\pi$ and action function $\epsilon$ are defined as follows:
		
		$\pi(q_0) = \{\alpha_1\}$; \quad $\varepsilon_1(q_0) = \{r, a\}$; \quad $\varepsilon_2(q_0) = \{r, a\}$;
		
		$\pi(q_1) = \{\alpha_2\}$; \quad $\varepsilon_1(q_1) = \{r, a\}$; \quad $\varepsilon_2(q_1) = \{r, a\}$;
		
		$\pi(q_2) = \{\beta_1\}$; \quad $\varepsilon_1(q_2) = \{r, w\}$; \quad $\varepsilon_2(q_2) = \{r, w\}$;
		
		$\pi(q_3) = \{\beta_2\}$; \quad $\varepsilon_1(q_3) = \{r, w\}$; \quad $\varepsilon_2(q_3) = \{r, w\}$;
		
		$\pi(q_4) = \{\epsilon\}$; \quad~~ $\varepsilon_1(q_4) = \{r\}$; \quad~~~ $\varepsilon_2(q_4) = \{r\}$;
		
		%		\vspace*{1mm}
		That is, in state $q_0,q_1$, agent 1 and agent 2 respectively have the priority to apply for controlling the shared storage; in state $q_2,q_3$, agent 1 and agent 2 respectively have been authorized to control the shared storage; $q_4$ is an error state; action $a$ is ``applying for controlling the storage", action $r$ is ``executing routine tasks`" and action $w$ is ``Writing to the storage". E.g., $\varepsilon_1(q_0) = \{r, a\}$ means in state $q_0$, Agent 1 can choose to perform routine tasks or apply for controlling the storage.
		
		\item The state transition function can be depicted as the directed edges in Figure~\ref{fig:ex1-cgs}.
		\begin{figure}
			\centering
			\includegraphics[width=0.55\linewidth]{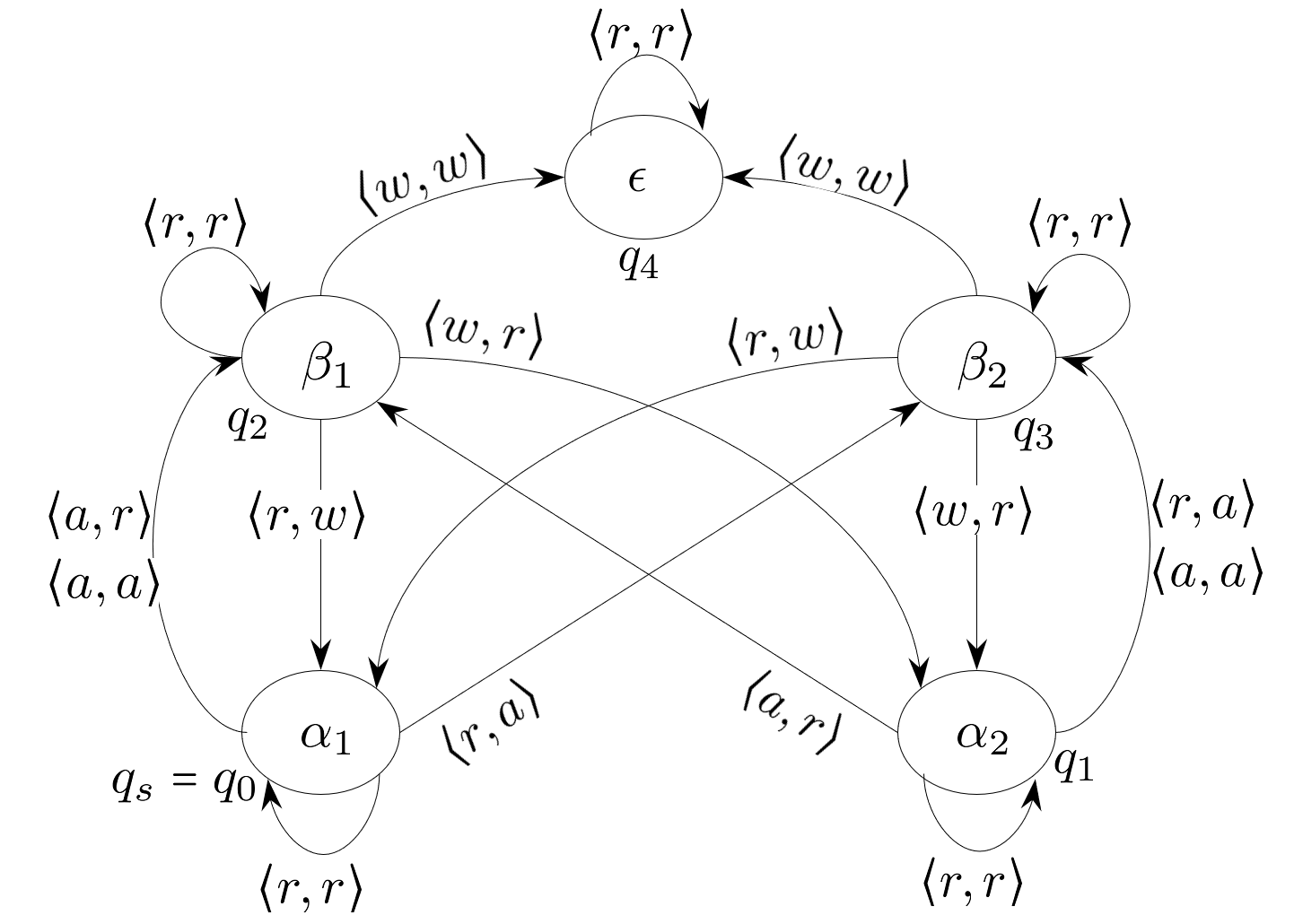}
			\caption{A cost-aware concurrent game structure (CCGS)}
			\label{fig:ex1-cgs}
		\end{figure}
		
		For example, in state $q_0$, agent 1 has the priority to apply for controlling the storage, that is, agent 1 can obtain the right to control the storage (via enforcing the system state to change to $q_2$ by choosing action $a$) no matter agent 2 chooses action $r$ or action $a$; While agent 2 can obtain the right to control the storage (let the system state change to $q_3$) only when agent 1 chooses action $r$ (agent 2 can choose action $a$ in this case); Moreover, when both of these two agents choose action $r$, the system state keeps the same. In state $q_2$, although agent $1$ has the right to control the storage, but agent 2 still can write to the storage. If only agent 1 writes in the storage, then the system state change to $q_1$ where agent $2$ has the priority to apply for controlling the storage; if only only agent 2 writes in the storage, then state change to $q_0$ where agent $1$ has the priority to apply for controlling the storage; when both of the two agents write in the storage, the system state transfers to $q_4$ since mutually exclusiveness is violated.
		
		\item The feature set can be represented as the Table~\ref{tb-fs}:
		
		According to lemma~\ref{lm-rl}, the above $\mathcal{F}$ actually specifies a valuation function $v_\mathcal{F}$.
		
		\item If $c(1) = 10$ and $c(2)=15$, that is, the unit cost of agent 1 and agent 2 are $10$ and $15$ respectively, and we pay each agent exactly their costs,  then some social laws and their values, costs and profits can be listed as Table ~\ref{tb-slu}.		
		
		For example, the structures obtained by implementing social laws $\eta_2$, $\eta_3$, $\eta_5$ and $\eta_7$ respectively on $S$ can be depicted in Figure~\ref{fig:fig-ex-cgs}. Among the $8$ social laws, the social law with highest value is $\eta_7$, while the social law with the highest profit is $\eta_3$; Although both $\eta_5$ and $\eta_7$ have higher values than social law $\eta_3$, their profits are lower than $\eta_3$;  $\eta_2$ and $\eta_7$ have the same profits although their values are different.
		
	\end{itemize}
	Our aim is to find out the social law $\eta^*$ that maximize the profit $g(\eta)$. But now notice that, $c(i)$ is only known by agent $i$ itself, what we can obtain is only $c(i)$'s probability density function $f_i$. In such a setting, it is obviously impossible for us to directly optimize $g(\eta)$. How to select the social law and pay the agents to obtain the best profit and at the same time keep the agents well-motivated becomes a challenge. 
	
	To overcome this challenge, based on the framework of algorithmic mechanism design, we can select the social law via running a procurement auction: firstly we announce a mechanism consists of an \emph{allocation function} $\hat{R}:X\rightarrow \mathcal{SL}_S$ and a \emph{payment function} $P_i:X\rightarrow \mathbb{R}^+$ for each agent $i$, then collect the bids (\emph{i.e.}, cost reports) $x$ from the agents, and finally select the social law $\eta=\hat{R}(x)$ and pay  $P_i(x)$ units of money to each agent $i$. Therefore, our paper focuses on deducing the functions $\hat{R}$ and $P_i$, which ensure biding truthfully being every agent's best choice (incentive compatibility), every agent getting a positive utility (individual rationality) and reliably obtaining the profit-optimal social law.
	
	\begin{figure}
		\centering
		\includegraphics[width=1.0\linewidth]{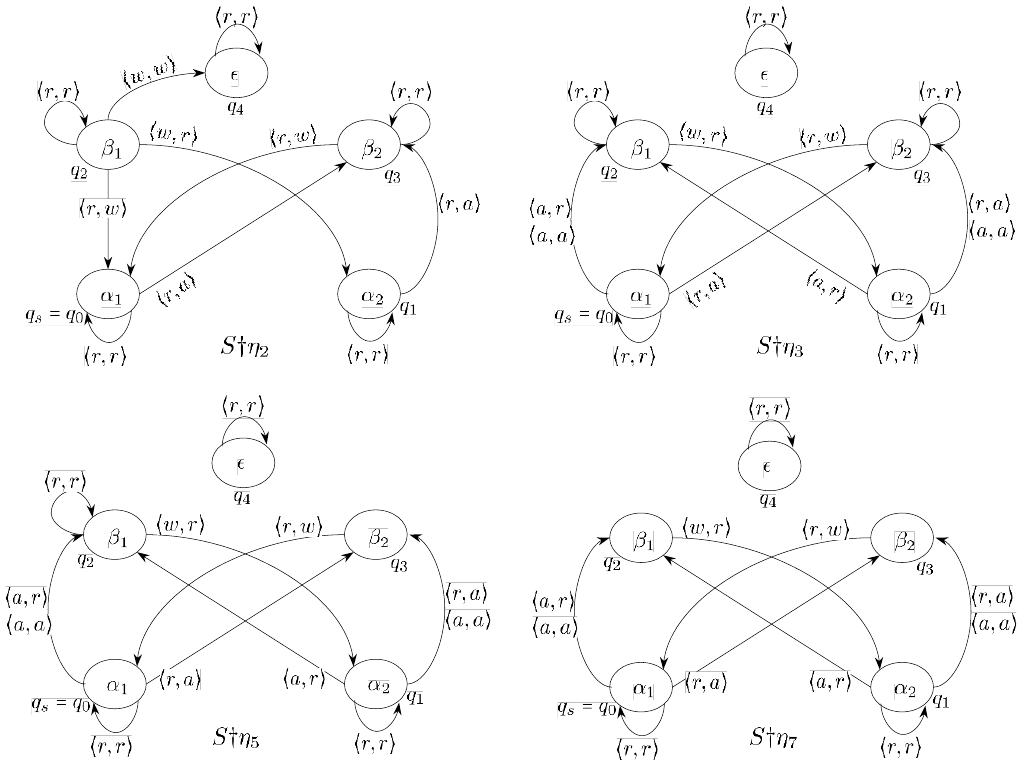}
		\caption{Structures obtained by implementing social laws}
		\label{fig:fig-ex-cgs}
		%		\vspace*{-3mm}
	\end{figure}
\end{ex}

\begin{ex}
	Given an agent set $Ag$, a proposition set $\Pi$, for any coalition $A_1,A_2,A_3 \subseteq Ag$ and propositions $p_1,p_2,p_3 \subseteq \Pi$, $\neg \llangle A_1\rrangle\Box (\llangle A_2\rrangle\bigcirc p_1  \vee \llangle A_3\rrangle p_2\mathcal{U} p_3 )$ is an ATL formula, the formulas in $cl(\neg \llangle A_1\rrangle\Box (\llangle A_2\rrangle\bigcirc p_1  \vee \llangle A_3\rrangle p_2\mathcal{U} p_3 ))$ can be represented as a tree-like structure depicted as Figure~\ref{fig:ex2}. 
	
	It is clear that,   
	\begin{itemize}
		\item the root node is the given formula, 
		\item all the leaf nodes are propositions in $\Pi$, and
		\item for any tree node $v$ and the corresponding formula $\varphi$,
		\begin{itemize}
			\item $cl(\varphi)$  is actually the set of all the nodes included in the sub-tree rooted at $v$;
			\item $sub(\varphi)$  is actually the set of all the nodes in the sub-trees which rooted at node $v$'s children (if $v$ isn't a leaf node) or the singleton set $\{v\}$ (if $v$ is a leaf node).	
		\end{itemize}		 
	\end{itemize}
\end{ex}

\begin{figure}
	\centering
	\includegraphics[width=0.35\linewidth]{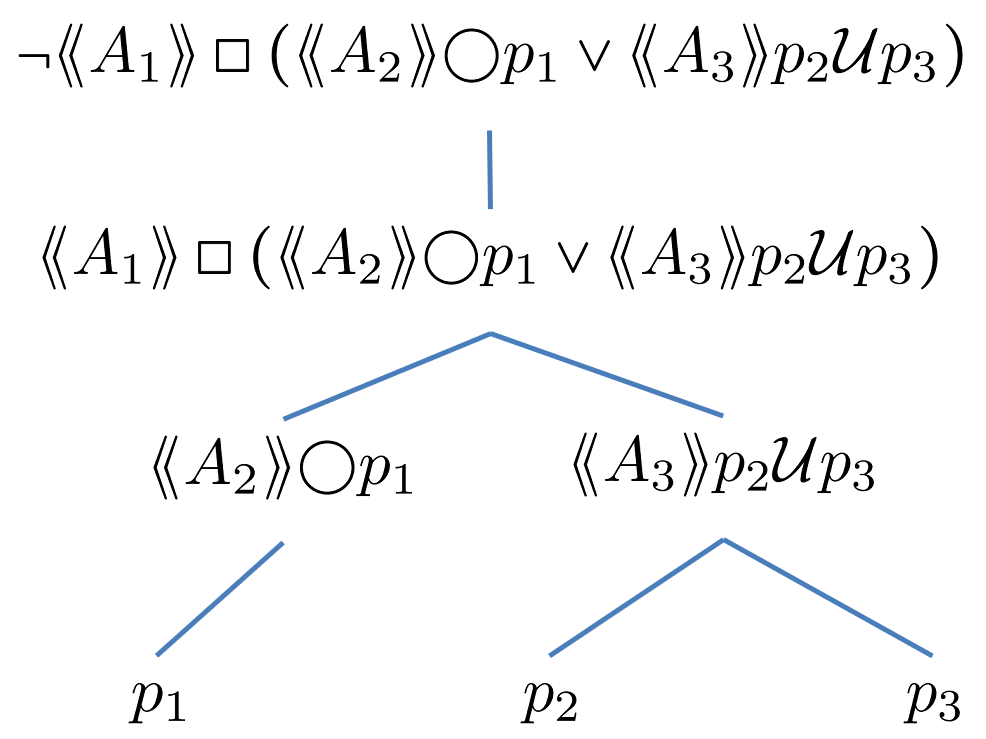}
	\caption{The closure of a formula and its underlying tree-like structure}
	\label{fig:ex2}
\end{figure}

\section{Results and Omitted Proofs}

\setcounter{thm}{0}

\begin{thm}(Alur et al. 1998) For two concurrent game structures $S$ and $S'$, if $S \leftrightarroweq S'$ and $q Z q'$, then for an arbitrary ATL function $\varphi$, we have $S, q\vDash \varphi~~\text{iff}~~S', q'\vDash \varphi.$
\end{thm}

\begin{lm}[Representation Lemma]\label{lm-rl} The following 2 items are equivalent:
	\begin{itemize}
		\item[1)] $\forall S_1, S_2$: $S \leftrightarroweq S'$ implies $v(S) = v(S')$;
		\item[2)] $\exists$ feature set $\mathcal{F} = \{(\varphi_1,c_1),...,(\varphi_k,c_k)\}$, where $k\in\mathbb{N},\forall j\in Ag: \varphi_j\in\mathcal{L}_{ATL}, c_j\in\mathbb{R}^+$, and for all $S\in\mathcal{W}_S$:
		$$v(S)= \sigma(\mathcal{F},S) = \sum_{(\varphi_j,c_j)\in\mathcal{F}; S,q_s\vDash\varphi_j}c_j.$$
	\end{itemize}
\end{lm}
\begin{proof}
	``$1)\Rightarrow 2)$": Since there are finite number of agents and states in the cocurrent game structure $S$, and the available actions of any agent in any state are finite, $\mathcal{SL}_S$ is a finite set of social laws, in which we can define a binary relation $\simeq$ as follows:
	$$\forall \eta_1,\eta_2\in \mathcal{SL}_S, \eta_1\simeq\eta_2~~\text{iff}~~S\dag \eta_1\leftrightarroweq S\dag \eta_2$$
	
	Obviously, the relation $\simeq$ satisfy reflexivity, symmetry and transitivity, and thus is an equivalence relation. We let $\mathcal{E} = \{ \eta_1,...,\eta_m\} $ be an equivalence class of the this equivalence relation, then for any social law $\eta\in\mathcal{SL}_S$ there must be a $\eta_j \in \mathcal{E} $ satisfying $S\dag \eta\leftrightarroweq S\dag \eta_j$; Moreover, for any  $\eta_j,\eta_l \in \mathcal{E} $, $S\dag \eta_j\leftrightarroweq S\dag \eta_l$ if and only if $j = l $. So, for any $j,l\in \{0,...,m\} $ and $j \neq l $,  $S\dag \eta_j$ and $S\dag \eta_l$ don't satisfy alternating bisimulation equivalence, therefore there exists a ATL formula $\psi_{j-l} $ satisfying:
	$$S\dag \eta_j, q_s \vDash \psi_{j-l}~~\text{and}~~S\dag \eta_l, q_s \nvDash \psi_{j-l}$$
	
	We can define a feature set $\mathcal{F} = \{(\varphi_1, c_1),...,(\varphi_m,c_m)\} $, where
	$$\forall 1\leq j\leq m:~~~~\varphi_j = \bigwedge_{l\neq j}\psi_{j-l} ~~~~ c_j = v(S\dag \eta_j)$$
	
	It is easy to see that $\varphi_j $ is a formula that is only satisfied by the initial state of $S\dag \eta_j$ but can't be satisfied by the initial state of any other structures in $\mathcal{W}_S$, that is, we have
	$$S\dag \eta_j, q_s \vDash \varphi_{l}~~\text{and}~~ j=l.$$ 
	It further follows that
	\begin{equation*}
		\begin{split}
			\forall \eta_j\in \mathcal{E}: v(S\dag \eta_j) = c_j = \!\!\sum_{(\varphi_{x},c_{x})\in \mathcal{F}: S\dag \eta_j,q_s\vDash \varphi_{x}}\!\!\!\!c_{x} = \sigma(\mathcal{F}, S\dag \eta_j)
		\end{split}
	\end{equation*}
	\begin{equation*}
		\begin{split}
			&\text{and}~~\forall \eta\nin \mathcal{E}: \exists \eta_j\in \mathcal{E},  v(S\dag\eta) = v(S\dag \eta_j) = c_j\\
			&= \sum_{(\varphi_{x},c_{x})\in \mathcal{F}: S\dag \eta_j,q_s\vDash \varphi_{x}}c_{x} = \sum_{(\varphi_{x},c_{x})\in \mathcal{F}: S\dag \eta,q_s\vDash \varphi_{x}}c_{x}\\
			&=\sigma(\mathcal{F}, S\dag \eta)
		\end{split}
	\end{equation*}
	
	So, for any $S\in \mathcal{W}_S$, we always have $v(S)=\sigma(\mathcal{F},S)$.\\

	$``1)\Leftarrow 2)":$ Assume that $\exists S_1,S_2\in \mathcal{W}_S$, $S_1\leftrightarroweq S_2$, and $v(S_1)\neq v(S_2)$. Since there is a $\mathcal{F} = \{(\varphi_1,c_1),...,(\varphi_k,c_k)\}$ satisfying:
	
	\begin{equation*}
		\begin{split}
			&v(S_1)=\sigma(\mathcal{F},S_1) = \sum_{(\varphi_j,c_j)\in\mathcal{F}; S_1,q_s\vDash\varphi_j}c_j;~~\text{and}\\ 
			&v(S_2)=\sigma(\mathcal{F},S_2) = \sum_{(\varphi_j,c_j)\in\mathcal{F}; S_2,q_s\vDash\varphi_j}c_j
		\end{split}
	\end{equation*}
	
	$v(S_1)\neq v(S_2)$ follows $$\sum_{(\varphi_j,c_j)\in\mathcal{F}; S_1,q_s\vDash\varphi_j}c_j \neq \sum_{(\varphi_j,c_j)\in\mathcal{F}; S_2,q_s\vDash\varphi_j}c_j$$
	
	That is, $S_1$ and $S_2$ satisfy different formula set in their initial state $q_s$, however this is impossible, since $S_1\leftrightarroweq S_2$.	
\end{proof}

\begin{lm}\label{lm-U2}
	A mechanism $\langle R,P\rangle$ is BNIC iff $\forall x_i\in X_i$:
	\begin{itemize}
		\item[1)]   $r_i(x_i)$ is monotone nonincreasing; and
		\item[2)] 	the expected payment to each agent satisfy:
		%		\vspace{-1mm}
		\begin{equation}
			\label{eq-pixipi0} p_i(x_i) = p_i(0) + x_i r_i(x_i) - \int_{0}^{x_i} r_i(t_i)dt_i\tag{14}
		\end{equation}.
	\end{itemize}	
\end{lm}
\begin{proof}
	
	First of all, the following equivalence relations follow trivially from the definition of BNIC and equations (9) and (10): 
	
	Mechanism $\langle H,P\rangle$ is BNIC 
	\begin{equation*}
		\begin{split}
			&\text{iff}~\forall i\in A, x_i, \gamma_i\in X_i:\bar{u}_i(x_i,x_i)\geq \bar{u}_i(x_i,\gamma_i)\\
			&\text{iff}~ \forall i\in A, x_i, \gamma_i\in X_i: \hat{u}_i(x_i)\geq p_i(\gamma_i) -r_i(\gamma_i) x_i\\ 
			&\text{iff}~ \forall i\in A, x_i, \gamma_i\in X_i:\hat{u}_i(x_i)-\hat{u}_i(\gamma_i)\geq  r_i(\gamma_i)(\gamma_i-x_i).
		\end{split}
	\end{equation*}

	Based on the above observation, we can prove both directions of this result as follows:
	
	``$\Rightarrow$": For all $i$, $x_i$ and $\gamma_i$, we have both 
	\begin{equation*}
		\begin{split}
			&\hat{u}_i(x_i)-\hat{u}_i(\gamma_i)\geq r_i(\gamma_i)(\gamma_i-x_i)~\text{and}\\
			&\hat{u}_i(\gamma_i)-\hat{u}_i(x_i)\geq r_i(x_i)(x_i-\gamma_i)
		\end{split}
	\end{equation*}

	Then we can obtain 
	\begin{equation}
		r_i(x_i)(x_i-\gamma_i)\leq \hat{u}_i(\gamma_i)-\hat{u}_i(x_i)\leq r_i(\gamma_i)(x_i-\gamma_i)
	\end{equation}
	It follows that $r_i(x_i)\leq r_i(\gamma_i)$ iff $x_i\geq \gamma_i$. Therefore,  $r_i(x_i)$ is a monotone nonincreasing function.
	
	To show the correctness of equation (\ref{eq-pixipi0}), we firstly divide the interval $[0,x_i]$ into $L$ intervals of length $\delta = \frac{x_i}{L}$. Denote by $y^k = (k+1)\delta$ the rightmost end of the $k$th interval, and by $x^k = k\delta$ its leftmost end. Let $x_i = y^k$ and $\gamma_i = x^k$,  then
	\begin{equation}\label{ieq-U}
		\begin{split}
			\sum_{k=0}^{L-1} r_i(y^k)(y^k-x^k)\leq \sum_{k=0}^{L-1} \hat{u}_i(x^k)-\hat{u}_i(y^k)\\
			\leq\sum_{k=0}^{L-1} r_i(x^k)(y^k-x^k)
		\end{split}
	\end{equation}	
	Noticing that, $y^k = x^{k+1}$ for all $0\leq k\leq L-1$, therefore
	\begin{equation}
		\sum_{k=0}^{L-1} \hat{u}_i(x^k)-\hat{u}_i(y^k) = \hat{u}_i(0) - \hat{u}_i(x_i)
	\end{equation}
	Both the left part and right part of inequality~(\ref{ieq-U}) are Rieman sums. By increasing $L$, $\delta$ gradually approaches $0$, both of the left part and right part of inequality~(\ref{ieq-U}) converge to $\int_{0}^{x_i} r_i(t_i) dt_i$. Therefore,
	\begin{equation}\label{eq-hatuihatui}
		\hat{u}_i(0) - \hat{u}_i(x_i) = \int_{0}^{x_i} r_i(t_i)dt_i
	\end{equation}
	Moreover, equation~(10) follows 
	\begin{equation}\label{eq-hatui0}
		\hat{u}_i(0) = p_i(0)
	\end{equation}
	Finally, the equation in item 2) of this lemma follows by combining the equations~(10)(\ref{eq-hatuihatui})and (\ref{eq-hatui0}).\\
	
	\noindent ``$\Leftarrow$": By equation~(10), equation (\ref{eq-pixipi0}) is equivalent to
	\begin{equation}\label{eq-uixi}
		\hat{u}_i(x_i) = p_i(0) - \int_{0}^{x_i} r_i(t_i)dt_i
	\end{equation}
	So, for all $\gamma_i\in X_i$
	\begin{equation}
		\hat{u}_i(x_i)  - \hat{u}_i(\gamma_i) = \int_{x_i}^{\gamma_i}r_i(t_i)dt_i
	\end{equation}
	Since $h_i$ is monotone nonincreasing, 
	\begin{equation}
		\int_{x_i}^{\gamma_i}r_i(t_i)dt_i\geq (\gamma_i-x_i)r_i(\gamma_i)
	\end{equation}
	then 
	$\hat{u}_i(x_i)  - \hat{u}_i(\gamma_i) \geq (\gamma_i-x_i)r_i(\gamma_i)$. Therefore, according to the observation we get in the beginning of this proof, $\langle H,P\rangle$ is BNIC.
\end{proof}
%\vspace{-5mm}
%Note that, the proofs for lemma~\ref{{lm-U2} and the following lemmas~\ref{lm-ep0}, lemma~\ref{lm-g0}, lemma~\ref{lm-qni}, and theorem~\ref{thm-bio} are included in the appendix.

	\begin{lm}\label{lm-ep0}
		A mechanism $\langle R,P \rangle$ is BNIC only if 
		%	\vspace{-3mm}
		\begin{equation}\label{eq-Ep2}
			\mathbb{E}_{x\in X}[\sigma(x)]\! =\!\int_X\!\Big(v_{\mathcal{F}}(S\dag\hat{R}(x))- \sum_{i\in A} \lambda_i(x_i)R_i(x)\Big)f(x)dx\tag{15}
		\end{equation}
		%		\vspace{-6mm}
		$$- \sum_{i\in A}\Big(p_i(0)-\int_{0}^{+\infty} r_i(t_i)dt_i\Big)$$
	\end{lm}
	\begin{proof}
		By Lemma~\ref{lm-U2},
		\begin{equation}
			p_i(x_i) = p_i(0) + x_ir_i(x_i) - \int_{0}^{x_i} r_i(t_i)dt_i
		\end{equation} 
		Therefore, by equation~(8) we can further obtain
		\begin{equation}
			\begin{split}
				\int_{X} P_i(x)f(x)dx &= \int_{0}^{+\infty}\!\!p_i(x_i)f(x_i)dx_i\\
				&= p_i(0) + \int_{0}^{+\infty}\!\! x_ir_i(x_i)f_i(x_i)dx_i\\
				&\quad- \int_{0}^{+\infty}\!\!\int_{0}^{x_i}\!\! r_i(t_i)dt_i f(x_i) dx_i
			\end{split}
		\end{equation}	
		
		Since
		\begin{equation}
			\int_{0}^{+\infty}\!\!\int_{0}^{x_i} r_i(t_i)dt_i f(x_i) dx_i = \int_{0}^{+\infty} r_i(t_i) (1-F_i(t_i))dt_i
		\end{equation}
		can be obtained by changing the order of integration. So,
		\begin{equation}\label{eq-int}
			\begin{split}
				\int_{X} P_i(x)f(x)dx &= p_i(0) - \int_{0}^{+\infty}r_i(x_i)dx_i\\
				&+ \int_{X}(x_i + \frac{F_i(x_i)}{f_i(x_i)}) R_i(x)f(x)dx\\
				&=p_i(0) - \int_{0}^{+\infty}r_i(x_i)dx_i\\
				&+ \int_{X}\lambda_i(x_i) R_i(x)f(x)dx
			\end{split}
		\end{equation}
		Finally, equation~(\ref{eq-Ep2}) can be obtained by combining equations~(\ref{eq-int}) and (13). 
	\end{proof} 
	
	\begin{lm}\label{lm-g0}
		A mechanism $\langle H,P\rangle$ is BNIC and IR only if	
		\begin{equation}\label{ieq-ph}
			\forall i\in A: p_i(0) - \int_{0}^{+\infty}r_i(t_i)dt_i \geq 0\tag{16}
		\end{equation}
	\end{lm}
	\begin{proof}
		$\langle H,P\rangle$ is BNIC follows equation (\ref{eq-uixi}), since we have proved that it implies equation (\ref{eq-pixipi0}) in lemma~\ref{lm-U2}.	Then by inequality (11), we have $$\forall i\in A: \hat{u}_i(+\infty) = p_i(0) - \int_{0}^{+\infty}r_i(t_i)dt_i \geq 0.$$
	\end{proof}

	\begin{lm}\label{lm-qni}
		The allocation function of $\mathcal{PO}$-$\mathcal{ASL}$ is monotone non-increasing.
	\end{lm}
	\begin{proof}
		Let $\hat{R}^*(x)$ be the social law that maximizes the function $$g(\hat{R}(x),x) = v_{\mathcal{F}}(S\dag\hat{R}(x)) - \sum_{i\in A} \lambda(x_i)R_i(x),$$ $\hat{R}'(x)$ be a social law that $R_k'(x)\geq R^*_k(x)$, and $\gamma = (\gamma_k,x_{-k})$ where $\gamma_k\geq x_k$. Suppose $R^*_k(x_k, x_{-k})$ is a increasing function of $x_k$, that is, when the bid vector changes to $\gamma$, the social law $\hat{R}'(x)$ becomes the one maximize the objective function. Therefore, $g(\hat{R}'(x),\gamma)\geq g(\hat{R}^*(x),\gamma)$. So,
		\begin{equation}\label{ieq-QQ}
			\begin{split}
				&v_{\mathcal{F}}(S\dag\hat{R}'(x)) - \sum_{i\in A} \lambda(\gamma_i)R'_i(x)\\
				&\geq v_{\mathcal{F}}(S\dag\hat{R}^*(x)) - \sum_{i\in A} \lambda(\gamma_i)R^*_i(x)
			\end{split}
		\end{equation}
		Since $\gamma_i = x_i$ for all $i\neq k$, inequality~(\ref{ieq-QQ}) is equivalent to
		\begin{equation}\label{ieq-QQ2}
			\begin{split}
				&g(\hat{R}'(x),x) +  (\lambda(x_k)-\lambda(\gamma_k))R'_k(x)\\
				&\geq g(\hat{R}^*(x),x)+ (\lambda(x_k)-\lambda(\gamma_k))R^*_k(x) 
			\end{split}
		\end{equation}
		By the regularity of the space $X$, we can further obtain $$g(\hat{R}'(x),x)\geq g(\hat{R}^*(x),x),$$ which contradicts the fact that $\hat{R}^*(x)$ is the social law with highest profit in this case. So, for each $k$, $R^*_k(x_k, x_{-k})$ and further $h^*_k(x_k)$ is a monotone non-increasing function of $x_k$.
	\end{proof}

	\begin{thm}\label{thm-bio}
		$\mathcal{PO}$-$\mathcal{ASL}$ is BNIC, IR, and maximizes the expected profit within all BNIC and IR mechanisms.
	\end{thm}
	\begin{proof}
		Since 
		\begin{equation*}
			\begin{split}
				p^*_i(x_i) &= \int_{X_{-i}}P^*_i(x_i,x_{-i})f_{-i}(x_{-i})dx_{-i}\\
				&=x_ih^*_i(x_i) + \int_{0}^{+\infty}h^*_i(t_i)dt_i - \int_{0}^{x_i}h^*_i(t_i)dt_i, 
			\end{split}		
		\end{equation*}
		we can obtain
		\begin{equation}\label{eq-m0}
			p^*_i(0) = \int_{0}^{+\infty}h^*_i(t_i)dt_i
		\end{equation}
		and further 
		\begin{equation}
			p^*_i(x_i) = p^*_i(0) + x_ih^*_i(x_i) - \int_{0}^{x_i} h^*_i(t_i)dt_i.
		\end{equation}
		Moreover, according to lemma~\ref{lm-qni}, $h^*_i$ is monotone non-increasing. So, by lemma~\ref{lm-U2},	$\mathcal{PO}$-$\mathcal{ASL}$ is BNIC.
		
		Then, by equations (\ref{eq-uixi}) and (\ref{eq-m0}), $\forall x_i\in X_i:$ 
		\begin{equation}
			\hat{u}^*_i(x_i)  = \int_{0}^{+\infty}h^*_i(t_i)dt_i - \int_{0}^{x_i}h^*_i(t_i)dt_i \geq 0
		\end{equation}
		
		So, $\mathcal{PO}$-$\mathcal{ASL}$ is also IR.
		\vspace*{2mm}
		
		Suppose there is a BNIC and IR mechanism $\mathcal{M'} = (H',P')$ which achieve higher expected profit than the $\mathcal{PO}$-$\mathcal{ASL}$ mechanism $\mathcal{M^*}=(H^*,P^*)$. By lemma~\ref{lm-U2},  $\mathcal{M'}$ is IR follows
		\begin{equation}\label{ieq-m2}
			\forall i\in A: p'_i(0) - \int_{0}^{+\infty}h'_i(t_i)dt_i\geq 0
		\end{equation}
		Therefore, by lemma~\ref{lm-ep0}, equality (\ref{eq-m0}) and inequality (\ref{ieq-m2}), $\mathbb{E}_{x\in X}[\sigma_{\mathcal{M'}}(x)] > \mathbb{E}_{x\in X}[\sigma_{\mathcal{M^*}}(x)]$ implies
		\begin{equation}
			\begin{split}
				&\int_X\Big(v_{\mathcal{F}}(S\dag\hat{R}'(x)) - \sum_{i\in A} \lambda(x_i) R'_i(x)\Big)f(x)dx\\
				&> \int_X\Big(v_{\mathcal{F}}(S\dag\hat{R}^*(x)) - \sum_{i\in A} \lambda(x_i) R^*_i(x)\Big)f(x)dx	
			\end{split}		
		\end{equation}	
		which contradicts the fact that for all $x\in X$, $R^*(x)$ maximizes $v_{\mathcal{F}}(S\dag\hat{R}(x)) - \sum_{i\in A} \lambda(x_i) R_i(x)$.
	\end{proof}

	\begin{thm}(Archer and Tardos 2001; Hartline 2006)\label{thm-truthful}
		A single-parameter procurement auction is truthful if and only if for any agent $i$ and bids of other
		agents $x_{-i}$ fixed, 
		\begin{itemize}
			\item $R_i(x_i,x_{-i})$ is monotone non-increasing.
			\item $P_i(x_i,x_{-i}) = R_i(x_i,x_{-i})x_i+\int_{x_i}^{+\infty}R_i(t_i,x_{-i})dt_i$
		\end{itemize}
	\end{thm}
	
	\begin{cor}
		$\mathcal{PO}$-$\mathcal{ASL}$ is truthful.
	\end{cor}	
	\begin{proof}
		Trivial, according to lemma~\ref{lm-qni}, theorem~\ref{thm-truthful} and the specification of mechanism $\mathcal{PO}$-$\mathcal{ASL}$.
	\end{proof}
	
	\begin{lm}\label{lm-nocatch}
		If $g(\eta,(x_i,x_{-i})) \geq g(\eta',(x_i,x_{-i}))$ and $\sum_{q\in Q} |\eta_i(q)| \leq \sum_{q\in Q} |\eta'_i(q)|$, then $g(\eta,(t,x_{-i})) \geq g(\eta',(t,x_{-i}))$ for all $t\geq x_i$.
	\end{lm}
	\begin{proof}
		$g(\eta,(x_i,x_{-i})) \geq g(\eta',(x_i,x_{-i}))$ implies 
		\begin{equation}
			\begin{split}
				&v_{\mathcal{F}}(S\dag\eta) - \sum_{i\in A} (\lambda_i(x_i)\cdot\sum_{q\in Q} |\eta_i(q)|) \\
				&\geq v_{\mathcal{F}}(S\dag\eta') - \sum_{i\in A} (\lambda_i(x_i)\cdot\sum_{q\in Q} |\eta'_i(q)|)
			\end{split}
		\end{equation}
		So, we have 
		\begin{equation} 
			\begin{split} 
				&v_{\mathcal{F}}(S\dag\eta') -\!\!\!\! \sum_{j\in A\setminus \{i\}}\!\!\! (\lambda_j(x_j)\cdot\!\sum_{q\in Q} |\eta'_j(q)|)-  \lambda_i(x_i)\cdot\!\sum_{q\in Q} |\eta'_i(q)|\\
				&\geq v_{\mathcal{F}}(S\dag\eta) -\!\!\!\! \sum_{j\in A\setminus \{i\}}\!\!\! (\lambda_j(x_j)\cdot\!\sum_{q\in Q} |\eta_j(q)|) -  \lambda_i(x_i)\cdot\!\sum_{q\in Q} |\eta_i(q)|		
			\end{split}  
		\end{equation}
		Therefore, for all $t\geq x_i$, we can further obtain
		\begin{equation} 
			\begin{split} 
				v_{\mathcal{F}}(S\dag\eta') &- \sum_{j\in A\setminus \{i\}} (\lambda_j(x_j)\cdot\sum_{q\in Q} |\eta'_j(q)|)\\
				&-(v_{\mathcal{F}}(S\dag\eta) - \!\!\!\!\!\!\sum_{j\in A\setminus \{i\}} (\lambda_j(x_j)\cdot\sum_{q\in Q} |\eta_j(q)|))\\  		
				&\leq \lambda_i(x_i)\cdot(\sum_{q\in Q} |\eta'_i(q)|- \sum_{q\in Q} |\eta_i(q)|)\\
				&\leq \lambda_i(t)\cdot(\sum_{q\in Q} |\eta'_i(q)|- \sum_{q\in Q} |\eta_i(q)|)
			\end{split}  
		\end{equation}
		Note that, the last inequality above is due to the fact that $\lambda_i$ is non-decreasing function of $x_i$ for every $i$, and $\sum_{q\in Q} |\eta'_i(q)|- \sum_{q\in Q} |\eta_i(q)|\geq 0$, and then we can transform the above inequality to the following form:
		\begin{equation}
			\begin{split}
				\forall t\geq x_i: &v_{\mathcal{F}}(S\dag\eta) - \sum_{j\in A\setminus \{i\}} (\lambda_j(x_j)\cdot\sum_{q\in Q} |\eta_j(q)|)\\ 
				&- \lambda_i(t)\cdot\sum_{q\in Q} |\eta_i(q)|\\
				&\geq v_{\mathcal{F}}(S\dag\eta') -\!\!\! \sum_{j\in A\setminus \{i\}} (\lambda_j(x_j)\cdot\sum_{q\in Q} |\eta'_j(q)|)\\ 
				&- \lambda_i(t)\cdot\sum_{q\in Q} |\eta'_i(q)|
			\end{split}
		\end{equation}
		That is, we have $\forall t\geq x_i: g(\eta,(t,x_{-i})) \geq g(\eta',(t,x_{-i}))$.         
	\end{proof}

	\begin{lm}\label{lm-eta}
		$\forall t_i \in [0,+\infty) : \eta^{(i,n,(t_i,x_{-i}))} = \eta^{(i,n,(0,x_{-i}))}$. 
	\end{lm}
	\begin{proof}
		\begin{equation*}
			\begin{split}
				&\forall t_i \in [0,+\infty) : \eta^{(i,n,(t_i,x_{-i}))} = \arg\!\!\!\! \max_{\eta\in \mathcal{SL}^{(i,n)}_S} g(\eta, (t_i,x_{-i}))\\
				&= \arg\!\!\!\!\! \max_{\eta\in \mathcal{SL}^{(i,n)}_S}\!\! \Big(v_{\mathcal{F}}(S\dag\eta) -\!\!\!\!\!\!\! \sum_{j\in A\setminus \{i\}}\!\!\!\!\!\! (\lambda_j(x_j)\!\cdot\!\!\!\sum_{q\in Q}\!\!|\eta_j(q)|)\! - \! \lambda_i(t_i)\!\cdot\!\!\!\sum_{q\in Q}\!\! |\eta_i(q)|\Big)\\
				&= \arg\!\!\!\! \max_{\eta\in \mathcal{SL}^{(i,n)}_S} \Big(v_{\mathcal{F}}(S\dag\eta) -\!\!\!\! \sum_{j\in A\setminus \{i\}}\!\!\!\! (\lambda_j(x_j)\cdot\sum_{q\in Q} |\eta_j(q)|) -  n\lambda_i(t_i)\Big)\\
				&= \arg\!\!\!\! \max_{\eta\in \mathcal{SL}^{(i,n)}_S} \Big(v_{\mathcal{F}}(S\dag\eta) -\!\!\!\! \sum_{j\in A\setminus \{i\}}\!\!\!\! (\lambda_j(x_j)\cdot\sum_{q\in Q} |\eta_j(q)|) -  n\lambda_i(0)\Big)\\ 
				&\qquad (\text{since both}~n\lambda_i(t_i)~\text{and}~n\lambda_i(t_i)~\text{are constants})\\
				&= \arg\!\!\!\! \max_{\eta\in \mathcal{SL}^{(i,n)}_S} g(\eta, (0,x_{-i})) = \eta^{(i,n,(0,x_{-i}))}
			\end{split}
		\end{equation*}
	\end{proof}
	
	\begin{lm}\label{lm-restrict}
		$$\forall t_i \in [x_i,+\infty) : \arg\max_{\eta\in \mathcal{SL}_S} g(\eta, (t_i,x_{-i})) = $$
		$$ \arg\max_{\eta \in \{\eta^{(i,n,x_{-i})}~|~ n\leq R^*_i(x_i,x_{-i}) \}} g(\eta, (t_i,x_{-i})).$$ 
	\end{lm}
	\begin{proof}
		Let $\eta^* = \hat{R}^*_i(x_i,x_{-i})$, and $\eta'$ be an arbitrary social law satisfying
		\begin{equation}
			\sum_{q\in Q} |\eta^*_i(q)| \leq \sum_{q\in Q} |\eta'_i(q)|
		\end{equation}
		Since $\eta^*$ is the dominant social law under the bid profile $(x_i,x_{-i})$, we have  $g(\eta^*,(x_i,x_{-i})) \geq g(\eta',(x_i,x_{-i}))$. Therefore, according to lemma~\ref{lm-nocatch}, 
		\begin{equation}
			\forall t_i\geq x_i: g(\eta^*,(t,x_{-i})) \geq g(\eta',(t,x_{-i}))
		\end{equation}
		Therefore, let $m = R^*_i(x_i,x_{-i})$, then $\forall t_i \in [x_i,+\infty) :$
		\begin{equation*}
			\begin{split}
				&\arg\max_{\eta\in \mathcal{SL}_S} g(\eta, (t_i,x_{-i})) = \arg\!\!\!\!\!\max_{\eta\in \mathcal{SL}^{(i,0)}_S\cup ...\cup \mathcal{SL}^{(i,m)}_S}\!\!\!\!\!\!\!\!\! g(\eta, (t_i,x_{-i}))\\
				&\in \{\arg\max_{\eta\in \mathcal{SL}^{(i,0)}_S}   g(\eta, (t_i,x_{-i})), ...,  \arg\!\!\!\!\max_{\eta\in \mathcal{SL}^{(i,m)}_S}\!\!\!\!   g(\eta, (t_i,x_{-i}))\}
			\end{split}
		\end{equation*}
		$$= \{\eta^{(i,0,(x_i,x_{-i}))}, ..., \eta^{(i,R^*_i(x_i,x_{-i}),(x_i,x_{-i}))}\} (\text{by equation~19})$$
		$$= \{\eta^{(i,0,x_{-i})}, ..., \eta^{(i,R^*_i(x_i,x_{-i}),x_{-i})}\} \quad\quad\quad\quad (\text{by lemma~\ref{lm-eta}})$$
		So, since $ \arg\max_{\eta\in \mathcal{SL}_S} g(\eta, (t_i,x_{-i}))$ always be the dominant social law under bid profile $(t_i,x_{-i})$,  we have $\forall t_i \in [x_i,+\infty) :$
		$$ \arg\max_{\eta\in \mathcal{SL}_S} g(\eta, (t_i,x_{-i})) = \arg\!\!\!\!\!\!\!\!\max_{\eta \in \{\eta^{(i,n,x_{-i})}~|~ n\leq m \}} \!\!\!\!\!\!\!\!g(\eta, (t_i,x_{-i})).$$    
	\end{proof}  
	
	\begin{lm}\label{lm-vnt}
		$v(n,t) = v_n - n\lambda_i(t) $.
	\end{lm}
	\begin{proof}
		Let $\eta = \eta^{(i,n,x_{-i})}$, then we have 
		$$v(n,t) = v_{\mathcal{F}}(S\dag\eta) - \sum_{i\in A\setminus\{i\}} (\lambda_i(x_i)\cdot\sum_{q\in Q} |\eta_i(q)|) - n\lambda_i(t)$$
		$$= v(n,0) - n\lambda_i(t) = v_n - n\lambda_i(t)$$
	\end{proof}

	\begin{thm}\label{thm-turning}
		$\forall j\geq 1:$ If $n_{j-1} > 0$, then we have 
		$$p_j = \min_{0 \leq m < n_{j-1}} \lambda^{-1}_i(\frac{v_{n_{j-1}} - v_m}{n_{j-1} -m})$$
		$$n_j = \arg \min_{0 \leq m < n_{j-1}} \lambda^{-1}_i(\frac{v_{n_{j-1}} - v_m}{n_{j-1} -m})$$
		else we have $p_j = +\infty$, $n_j = 0$.
	\end{thm}
	\begin{proof} There are the following two cases:
		
		Case 1 ($n_{j-1} > 0$):
		By lemma~\ref{lm-restrict}, since $R^*_i(p_{j-1},x_{-i}) = n_{j-1}$, we have $\forall t_i \in [p_{j-1},+\infty) :$	
		$$ \arg\max_{\eta\in \mathcal{SL}_S} g(\eta, (t_i,x_{-i})) =  \arg\!\!\!\!\!\!\!\!\max_{\eta \in \{\eta^{(i,n,x_{-i})}~|~ n\leq n_{j-1} \}}\!\!\!\!\!\!\!\! g(\eta, (t_i,x_{-i}))$$
		That is, when agent $i$'s bid varies gradually from $p_{j-1}$ to $+\infty$, the selected social law is always the one in the set $\{\eta^{(i,n_{j-1},x_{-i})}, ..., \eta^{(i,0,x_{-i})}\}$ with highest objective function value.
		Moreover, when $t_i = p_{j-1}$, the dominant social law is $\eta^{(i,n_{j-1},x_{-i})}$. Therefore, for a social law $\eta^{(i,m,x_{-i})}$ where $0 \leq m < n_{j-1}$ in the above set, there must be a point $t = x_m$ where $v(m,t)$ catch up with $v(n_{j-1},t)$, and according to lemma~\ref{lm-vnt}, we have
		\begin{equation}
			v_{n_{j-1}} - n\lambda_i(x_m) = v_m - m\lambda_i(x_m)
		\end{equation}
		And therefore 
		\begin{equation}
			x_m = \lambda_i^{-1}(\frac{v_{n_{j-1}} - v_m}{n_{j-1} -m})
		\end{equation}
		Since we always select the dominant social law, as agent $i$'s bid increases gradually from $p_{j-1}$, we selected another social law as soon as we meet the first $x_m$ above, so we have 
		\begin{equation}
			p_j = \min_{0 \leq m < n_{j-1}}x_m=\min_{0 \leq m < n_{j-1}} \lambda^{-1}_i(\frac{v_{n_{j-1}} - v_m}{n_{j-1} -m})
		\end{equation}
		\begin{equation}
			n_j = \arg \min_{0 \leq m < n_{j-1}} \lambda^{-1}_i(\frac{v_{n_{j-1}} - v_m}{n_{j-1} -m})
		\end{equation}
		
		Case 2 ($n_{j-1} = 0$): According to definition, $j-1\geq k$, and so $j\geq k+1 >k$, which follows $p_{j} = +\infty$, and $n_{j} = 0$.    
	\end{proof}

	\begin{thm}\label{thm-tp}
		Let $k\in \mathbb{N}$ be the first position such that $n_k = 0$, then the payment to agent $i$ should be 
		
		$$P^*_i(x_i,x_{-i})  = \sum_{1\leq i\leq k} (n_{i-1} - n_i)p_i $$%= n_0 p_1 +\sum_{2\leq i\leq k} n_{i-1}(p_i - p_{i-1})
	\end{thm} 
	\begin{proof}
		According the specification of $\mathcal{PO}$-$\mathcal{ASL}$:
		\begin{equation*}
			\begin{split}
				&P^*_i(x) = R^*_i(x)x_i + \int_{x_i}^{+\infty}R^*_i(t_i,x_{-i})dt_i\\
				&=n_0 p_0 + \int_{p_0}^{p_1} n_0 dt_i+ \int_{p_1}^{p_2} n_1 dt_i +...+ \int_{p_{k-1}}^{p_k} n_{k-1} dt_i\\
				&=n_0 p_0 + n_0 (p_1- p_0)+ n_1(p_2-p_1) +...+ n_{k-1}(p_k - p_{k-1})\\
				&= n_0 p_1 + n_1(p_2-p_1) +...+ n_{k-1}(p_k - p_{k-1})\\
				&= n_0 p_1 + n_1(p_2-p_1) +...+ n_{k-1}(p_k - p_{k-1}) - n_kp_k\\
				&~~~~~\text{(Since $n_k = 0$)}\\
				&=\sum_{1\leq i\leq k} (n_{i-1} - n_i)p_i\\
			\end{split}
		\end{equation*}
	\end{proof}	
	
	\begin{cor}
		Algorithm 1 correctly computes the payment of mechanism $\mathcal{PO}$-$\mathcal{ASL}$.
	\end{cor}
	\begin{proof}
		Algorithm 1 first of all computes the dominant social law $\eta^*$ under the bid profile $(x_i,x_{-i})$ and $n_0$, \emph{i.e.}, the number of agent $i$'s actions that is restricted by $\eta^*$. Depending on $n_0$, there are only two possibilities as follows:
		
		Case 1 ($n_0=0$): Since $n_0 = R^*_i(x_i,x_{-i}) =0$, and obviously $R^*_i(t_i,x_{-i})\geq 0$ for all $t_i\geq x_i$, lemma~\ref{lm-qni} follows $R^*_i(t_i,x_{-i}) =0$ for all $t_i\geq x_i$. So, according to the specification of mechanism $\mathcal{PO}$-$\mathcal{ASL}$, we have $P^*_i(x) = 0$, which coincides with the result obtained by algorithm 1.
		
		Case 2 ($n_0 > 0$): Algorithm 1 computes the turning points according to theorem~\ref{thm-turning}, and then computes the payment according to theorem~\ref{thm-tp}, and therefore the correctness of the payment obtained by algorithm 1 is guaranteed by the correctness of the above two theorems.
	\end{proof}

	\begin{lm}\label{lm-cdsl}
		Both {\sc Dominant Social Law} and {\sc Dominant $(i,n)$-Social Law} are FP$^{NP}$-complete.
	\end{lm}
	\begin{proof}
		1) For {\sc Dominant Social Law}, the associated decision problem {\sc Dominant Social Law(d)} is ``whether there is a social law $\eta\in \mathcal{SL}_S$ achieves the value of $g(\eta,x)$ at least $l\in \mathbb{R}^+$". {\sc Dominant Social Law(d)} is trivially in $NP$ since we can guess a social law and verify it in polynomial time. Therefore, {\sc Dominant Social Law} is in FP$^{NP}$; Moreover, we can show that  {\sc Dominant Social Law} is FP$^{NP}$-hard by reducing {\sc Max Weight Sat}, which is well known to be FP$^{NP}$-compete, to it: An instance of {\sc Max Weight Sat} is given by a set of propositions $\{\psi_1,...,\psi_n\}$ over a set of Boolean variables $\{\phi_1,...,\phi_m\}$ together with integer weights $w_1,...,w_n$ for each proposition respectively. The aim is to find the valuation $v^*:\{\phi_1,...,\phi_m\}\rightarrow \{0,1\}$ that maximizes the sum of weights of propositions satisfied by the valuation. 
		
		We can construct a structure shown as Figure~\ref{fig:lemma17}: There is only one agent in the system, that is, $Ag = \{1\}$; the state set is $Q = \{q_s,q_1,...,q_m\}$, the transition relation is depicted as the arrows in the figure, that is, $\delta(q_s, a_i) = q_i$;  The virtual value of agent 1's bid is $\lambda_1(x_i)=0$.
		
		The feature set is $\mathcal{F} = \{(\hat{\psi_1},w_1),...,(\hat{\psi_n},w_n)\}$, where each $\hat{\psi_i}$ is obtained from $\psi_i$ by substituting each proposition $\phi_j$ with the ATL formula $\llangle 1\rrangle\bigcirc\phi_j$.
		
		Now, actually we have constructed an instance of {\sc Dominant Social Law}, and the output of this instance is a social law $\eta^*\in \mathcal{SL}_S$ that maximize 
		\begin{equation}\label{eq-objlm10}
			g(\eta,x) = v_{\mathcal{F}}(S\dag\eta) - \lambda_1(x_1)\cdot\sum_{q\in Q} |\eta_1(q)|
		\end{equation}
		
		\begin{equation}
			=	\sum_{(\hat{\psi}_i,w_i)\in\mathcal{F},S\dag\eta,s_0\models\hat{\varphi}_i}w_i
		\end{equation}
		
		So, the following valuation function
		$$\forall i\in \{1,...,m\}: v^*(\phi_i)=
		\begin{cases}
			1 & a_i\in \eta^*\\
			0 & \text{else}
		\end{cases}$$
		maximizes 
		\begin{equation}
			\sum_{(\psi_i,w_i)\in\mathcal{F},v\models \psi_i}w_i
		\end{equation}
		
		Therefore, {\sc Dominant Social Law} is FP$^{NP}$-complete.	%~~\quad\qquad\quad$\Box$	
		\begin{figure}
			\centering
			\includegraphics[width=0.4 \linewidth]{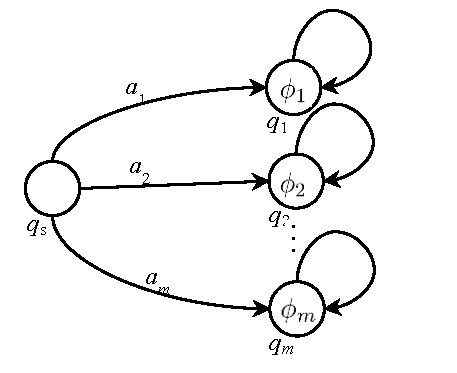}
			\caption{A structure for the proof of lemma~\ref{lm-cdsl}}
			\label{fig:lemma17}
		\end{figure}		
		
		2) For {\sc Dominant $(i,n)$-Social Law}, the associated decision problem is ``whether there is an $(i,n)$-social law $\eta\in \mathcal{SL}^{(i,n)}_S$ achieves the value of $g(\eta,x)$ at least $l\in \mathbb{R}^+$", and it is trivially in $NP$ since we can guess an $(i,n)$-social law and verify it in polynomial time. Therefore, {\sc Dominant $(i,n)$-Social Law} is in FP$^{NP}$; Moreover, we can show that it is FP$^{NP}$-hard by reducing {\sc Dominant Social Law}, which is FP$^{NP}$-complete according to lemma~\ref{lm-cdsl}, to this problem: in order to find out the dominant social law, we can first of all try to find out the dominant $(i,n)$-social laws for all $n\in [0,N]$, where $N$ is the number of agent $i$'s actions, and then select the dominant $(i,n)$-social law with highest objective function value. 
	\end{proof}
	
	\begin{thm}
		Mechanism $\mathcal{PO}$-$\mathcal{ASL}$ is FP$^{NP}$-complete.
	\end{thm}
	\begin{proof}
		The computing of mechanism $\mathcal{PO}$-$\mathcal{ASL}$ mainly involves solving $k+1$ instances of {\sc Dominant Social Law} and at most $kN$ instances of {\sc Dominant $(i,n)$-Social Law}, where $N$ is the upperbound of the action number of the agents, and so it is also FP$^{NP}$-complete, since according to lemma~\ref{lm-cdsl} both deciding the dominant social law and deciding the $(i,n)$-social law are FP$^{NP}$-complete.
	\end{proof}

	\begin{lm}\label{lm-x16}~		
		\begin{itemize}
			\item[1)] $\forall q, i, a$: $\ell(y^q_{i:a})=1$ iff $a\in\eta_\ell(i,q)$;
			\item[2)] $\forall q, A, \vec{m}_A$: $\ell(y^q_{A:\vec{m}_A})=1$ iff $\exists i\in A: \ell(y^q_{i:\vec{m}_A[i]}) = 1$;
			\item[3)] $\forall q, \llangle A\rrangle, \vec{m}_A,\vec{m}_{\bar{A}}$: $\ell(s^{q,\varphi}_{A:\vec{m}_A, \vec{m}_{\bar{A}}}) = 1$ iff $\ell(y^q_{\bar{A}:\vec{m}_{\bar{A}}})=1$ or $\ell(x^{\delta(q, (\vec{m}_A,\vec{m}_{\bar{A}}))}_\varphi) = 1$;
			\item[4)] $\ell(z^{q,\varphi}_{A:\vec{m}_A})=1$ iff $\forall \vec{m}_{\bar{A}}: $$\ell(s^{q,\varphi}_{A:\vec{m}_A, \vec{m}_{\bar{A}}}) = 1$;
			\item[5)] $\forall q, \varphi, A, \vec{m}_A$: $\ell(e^{q,\varphi}_{A:\vec{m}_A}) = 1$ iff $\ell(y^q_{A:\vec{m}_A})=0$ and $\ell(z^{q,\varphi}_{A:\vec{m}_A})=1$;
			\item[6)] $\forall q, \llangle A\rrangle\psi\mathcal{U}\chi$: $\ell(r^q_{\llangle A\rrangle \psi\mathcal{U}\chi}) = 1$ iff $\ell(x^q_{\psi})=1$ and $\exists \vec{m}_A: (\ell(y^q_{A:\vec{m}_A})=0$ and $\ell(z^{q,\llangle A\rrangle\psi\mathcal{U}\chi}_{A:\vec{m}_A})=1)$.
		\end{itemize}
	\end{lm}
	\begin{proof}
		\begin{itemize}
			
			\item [1)] It directly follows from the definition of $\eta_\ell(i,q)$;
			
			\vspace*{2mm}
			\item [2)] If $\ell(y^q_{A:\vec{m}_A})=1$, then by constraint~(30), $\ell(y^q_{A:\vec{m}_A}) \geq 1$, so $\exists i\in A: \ell(y^q_{i:\vec{m}_A[i]}) = 1$; 
			
			if $\exists i\in A: \ell(y^q_{i:\vec{m}_A[i]}) = 1$, then by constraint~(29), $\ell(y^q_{A:\vec{m}_A}) \geq 1$, so $\ell(y^q_{A:\vec{m}_A})=1$;
			
			\vspace*{2mm}
			\item [3)] If $\ell(s^{q,\varphi}_{A:\vec{m}_A, \vec{m}_{\bar{A}}}) = 1$, then by constraint~(34), we have $$\ell(y^q_{\bar{A}:\vec{m}_{\bar{A}}})+ \ell(x^{\delta(q, (\vec{m}_A,\vec{m}_{\bar{A}}))}_\varphi) \geq 1, $$ therefore $\ell(y^q_{\bar{A}:\vec{m}_{\bar{A}}})=1$ or  $\ell(x^{\delta(q, (\vec{m}_A,\vec{m}_{\bar{A}}))}_\varphi) = 1$; 
			
			If $\ell(y^q_{\bar{A}:\vec{m}_{\bar{A}}})=1$ or $\ell(x^{\delta(q, (\vec{m}_A,\vec{m}_{\bar{A}}))}_\varphi) = 1$, then by constraints (32)(33), we have $\ell(s^{q,\varphi}_{A:\vec{m}_A, \vec{m}_{\bar{A}}}) \geq 1$, so $\ell(s^{q,\varphi}_{A:\vec{m}_A, \vec{m}_{\bar{A}}}) = 1$;
			
			\vspace*{2mm}
			\item [4)] If $\ell(z^{q,\varphi}_{A:\vec{m}_A})=1$, then by constraint~(36), we have $$\forall \vec{m}_{\bar{A}}\in D_{\bar{A}}(q): \ell(s^{q,\varphi}_{A:\vec{m}_A, \vec{m}_{\bar{A}}}) \geq 1,$$ therefore $\forall \vec{m}_{\bar{A}}\in D_{\bar{A}}(q): $$\ell(s^{q,\varphi}_{A:\vec{m}_A, \vec{m}_{\bar{A}}}) = 1$; 
			
			If $\forall \vec{m}_{\bar{A}}\in D_{\bar{A}}(q): $$\ell(s^{q,\varphi}_{A:\vec{m}_A, \vec{m}_{\bar{A}}}) = 1$, then by constraint~(37), $\ell(z^{q,\varphi}_{A:\vec{m}_A})\geq 1$, so $\ell(z^{q,\varphi}_{A:\vec{m}_A}) = 1$;
			
			\vspace*{2mm}
			\item [5)] If $\ell(e^{q,\varphi}_{A:\vec{m}_A}) = 1$, then by constraint~(46), $\ell(z^{q,\varphi}_{A:\vec{m}_A})\geq 1$, so we have $\ell(z^{q,\varphi}_{A:\vec{m}_A}) = 1$; Moreover, by constraint ~(47), $\ell(y^q_{A:\vec{m}_A})\leq 0$, so we have  $\ell(y^q_{A:\vec{m}_A}) = 0$; 
			
			If $\ell(y^q_{A:\vec{m}_A})=0$ and $\ell(z^{q,\varphi}_{A:\vec{m}_A})=1$, then by constraint~(45), we have $\ell(e^{q,\varphi}_{A:\vec{m}_A}) \geq 1$, so $\ell(e^{q,\varphi}_{A:\vec{m}_A}) = 1$;
			
			\vspace*{2mm}
			\item [6)] If $\ell(r^q_{\llangle A\rrangle \psi\mathcal{U}\chi}) = 1$, then by constraint~(51), we have $\ell(x^q_{\psi})\geq 1$, therefore $\ell(x^q_{\psi}) = 1$ and by constraint~(52), we have $\sum_{\vec{m}_A\in D_A(q)} \ell(e^{q,\llangle A\rrangle\psi\mathcal{U}\chi}_{A:\vec{m}_A}) \geq 1$. It implies $\exists \vec{m}_A\in D_A(q): \ell(e^{q,\llangle A\rrangle\psi\mathcal{U}\chi}_{A:\vec{m}_A}) = 1$, so by lemma~\ref{lm-x16}.5, we can obtain $\ell(x^q_{\psi})=1$ and $\exists \vec{m}_A\in D_A(q): \ell(e^{q,\llangle A\rrangle\psi\mathcal{U}\chi}_{A:\vec{m}_A}) = 1$, and further by constraint~(53), we can deduce $\ell(r^q_{\llangle A\rrangle \psi\mathcal{U}\chi}) \geq 1$, so $\ell(r^q_{\llangle A\rrangle \psi\mathcal{U}\chi}) = 1$.
		\end{itemize}
	\end{proof}

	\begin{lm}\label{lm-c} For any $\varphi\in cl(\mathcal{F})$, $\ell(x^q_\varphi)=1$ iff $S\dag\eta_{\ell},q\vDash\varphi$.
	\end{lm}
	\begin{proof}
		For any $\varphi\in cl(\mathcal{F})$, we will try to prove this result by structural induction:
		
		\vspace*{2mm}		
		1) Case $\varphi\in cl(\mathcal{F})$, $p\in \Pi$: by constraints (38)(39), we can obtain $\ell(x^q_p)=1$ iff $p\in (\Pi\cap cl(\mathcal{F}))\cap \pi(q)$ iff $p\in \pi(q)$ iff $S\dag\eta_{\ell},q\vDash\varphi$;
		
		\vspace*{2mm}
		2) Case $\varphi = \neg \psi$: by constraint (40) we can obtain $\ell(x^q_{\neg \psi})=1$ iff $\ell(x^q_{\psi})=0$ (by inductive hypothesis) iff $S\dag\eta_{\ell},q\nvDash\varphi$ (by the ATL semantics) iff $S\dag\eta_{\ell},q\vDash \neg \varphi$;
		
		\vspace*{2mm}
		3) Case $\varphi = \psi\vee \chi$: by constraints (41)~(43), we can obtain $\ell(x^q_{\psi\vee \chi})=1$, iff $\ell(x^q_{\psi})=1$ or $\ell(x^q_{\chi})=1$ (by inductive hypothesis) iff 
		$S\dag\eta_{\ell},q\vDash \psi$ or $S\dag\eta_{\ell},q\vDash \chi$ (by the ATL semantics) iff $S\dag\eta_{\ell},q\vDash \psi\vee \chi$;
		
		\vspace*{2mm}
		4) Case $\varphi = \llangle A\rrangle\bigcirc \psi$: If $\ell(x^q_{\llangle A\rrangle \bigcirc\varphi})=1$, then by constraint (49) we can obtain $\sum_{\vec{m}_A\in D_A(q)} \ell(e^{q,\varphi}_{A:\vec{m}_A})\geq 1$, therefore $\exists \vec{m}_A\in D_A(q): \ell(e^{q,\varphi}_{A:\vec{m}_A})= 1$, by lemma~\ref{lm-x16}.5, we have $\exists \vec{m}_A\in D_A(q): \ell(e^{q,\varphi}_{A:\vec{m}_A})= 1$, and further by constraint (48), we can obtain $\ell(x^q_{\llangle A\rrangle \bigcirc\varphi})\geq 1$, and therefore $\ell(x^q_{\llangle A\rrangle \bigcirc\varphi})= 1$. So, now we have $\ell(x^q_{\llangle A\rrangle \bigcirc\varphi})=1$ 
		
		\vspace*{2mm}		
		\noindent iff $\exists \vec{m}_A\in D_A(q): (\ell(y^q_{A:\vec{m}_A})=0$ and $\ell(z^{q,\varphi}_{A:\vec{m}_A})=1)$;
		
		\vspace*{2mm}	
		\noindent iff $\exists \vec{m}_A\in D_A(q): ($$\forall i\in A: \vec{m}_A[i]\nin \eta_{\ell}(i,q)$ and $\forall \vec{m}_{\bar{A}}\in D_{\bar{A}}(q): $$\ell(s^{q,\varphi}_{A:\vec{m}_A, \vec{m}_{\bar{A}}}) = 1$) (by lemma~\ref{lm-x16}.1, lemma~\ref{lm-x16}.2 and lemma~\ref{lm-x16}.4) 
		
		\vspace*{2mm}	
		\noindent iff $\exists \vec{m}_A\in D_A(q): ($$\forall i\in A: \vec{m}_A[i]\nin \eta_{\ell}(i,q)$ and $\forall \vec{m}_{\bar{A}}\in D_{\bar{A}}(q): $$\ell(y^q_{\bar{A}:\vec{m}_{\bar{A}}})=1$ or $\ell(x^{\delta(q, (\vec{m}_A,\vec{m}_{\bar{A}}))}_\varphi) = 1$) (by lemma~\ref{lm-x16}.3);
		
		\vspace*{2mm}	
		\noindent iff $\exists \vec{m}_A\in D_A(q): ($$\forall i\in A: \vec{m}_A[i]\nin \eta_{\ell}(i,q)$ and $\forall \vec{m}_{\bar{A}}\in D_{\bar{A}}(q): $$\exists i\in \bar{A}: \ell(y^q_{i:\vec{m}_{\bar{A}}[i]}) = 1$ or $\ell(x^{\delta(q, (\vec{m}_A,\vec{m}_{\bar{A}}))}_\varphi) = 1$) (by lemma~\ref{lm-x16}.1, lemma~\ref{lm-x16}.2) ;
		
		\vspace*{2mm}	
		\noindent iff $\exists \vec{m}_A\in D_A(q): ($ $\forall i\in A: \vec{m}_A[i]\nin \eta_{\ell}(i,q)$ and $\forall \vec{m}_{\bar{A}}\in D_{\bar{A}}(q): $$\exists i\in \bar{A}: \ell(y^q_{i:\vec{m}_{\bar{A}}[i]}) = 1$ or $S\dag\eta_{\ell}, {\delta(q, (\vec{m}_A,\vec{m}_{\bar{A}}))} \vDash\varphi$) (by inductive hypothesis) 
		
		\vspace*{2mm}	
		\noindent iff $S\dag\eta_{\ell},q\vDash \llangle A\rrangle\bigcirc \psi$ (by the ATL semantics) ;
		
		\vspace*{2mm}
		5) Case $\varphi = \llangle A\rrangle\psi\mathcal{U}\chi$: If $\ell(x^q_{\llangle A\rrangle\psi\mathcal{U}\chi})=1$, then by constraint (56) we can obtain $\ell(x^q_{\chi})+ \ell(r^q_{\llangle A\rrangle \psi\mathcal{U}\chi}) \geq 1$, therefore $\ell(x^q_{\chi})=1$ or $\ell(r^q_{\llangle A\rrangle \psi\mathcal{U}\chi}) = 1$; 
		
		If $\ell(x^q_{\chi})=1$ or $\ell(r^q_{\llangle A\rrangle \psi\mathcal{U}\chi}) = 1$, then by constraints (54)(55), we can obtain $\ell(x^q_{\llangle A\rrangle\psi\mathcal{U}\chi})=1$;
		
		The above facts means $\ell(x^q_{\llangle A\rrangle\psi\mathcal{U}\chi})=1$ 
		
		\vspace*{2mm}
		\noindent iff $\ell(x^q_{\chi})=1$ or $\ell(r^q_{\llangle A\rrangle \psi\mathcal{U}\chi}) = 1$, 
		
		\vspace*{2mm}
		\noindent iff $\ell(x^q_{\chi})=1$ or ($\ell(x^q_{\chi})=1$ and $\exists \vec{m}_A\in D_A(q): (\ell(y^q_{A:\vec{m}_A})=0$ and $\ell(z^{q, \llangle A\rrangle\psi\mathcal{U}\chi}_{A:\vec{m}_A})=1)$)(by lemma~\ref{lm-x16}.6) , 
		
		\vspace*{2mm}
		\noindent iff $\ell(x^q_{\chi})=1$ or ($\ell(x^q_{\psi})=1$ and $\ell(x^q_{ \llangle A\rrangle \bigcirc \llangle A\rrangle\psi\mathcal{U}\chi})=1$) (by the proof of the above item 4) , 
		
		\vspace*{2mm}
		\noindent iff $S\dag\eta_{\ell},q\vDash \chi$ or ($S\dag\eta_{\ell},q\vDash \psi$ and $S\dag\eta_{\ell},q\vDash \llangle A\rrangle \bigcirc \llangle A\rrangle\psi\mathcal{U}\chi$) (by inductive hypothesis) , 
		
		\vspace*{2mm}
		\noindent iff $S\dag\eta_{\ell},q\vDash \chi\vee (\psi\wedge\llangle A\rrangle\bigcirc\llangle A\rrangle\psi\mathcal{U}\chi)$ (by the ATL semantics) ,  
		
		\vspace*{2mm}
		\noindent iff $S\dag\eta_{\ell},q\vDash \llangle A\rrangle\psi\mathcal{U}\chi$ (since $\llangle A\rrangle \psi\mathcal{U}\chi\leftrightarrow (\chi\vee (\psi\wedge\llangle A\rrangle\bigcirc\llangle A\rrangle\psi\mathcal{U}\chi))$ is valid);
		
		\vspace*{2mm}
		6) Case $\varphi = \llangle A\rrangle\Box\psi$: If $\ell(x^q_{\llangle A\rrangle\Box\psi})=1$, then by constraint (57) we can obtain $\ell(x^q_{\psi})\geq 1$, so we have $\ell(x^q_{\psi})= 1$, moreover by constraint (58) we have $\sum_{\vec{m}_A\in D_A(q)} \ell(e^{q,\llangle A\rrangle\Box\psi}_{A:\vec{m}_A}) \geq 1$, therefore $\exists \vec{m}_A\in D_A(q):$$\ell(e^{q,\llangle A\rrangle\Box\psi}_{A:\vec{m}_A}) = 1$; If $\ell(x^q_{\psi})=1$ and $\exists \vec{m}_A\in D_A(q):$$\ell(e^{q,\llangle A\rrangle\Box\psi}_{A:\vec{m}_A}) = 1$, then by constraint (59) we have $\ell(x^q_{\llangle A\rrangle \Box\psi}) \geq 1$, so $\ell(x^q_{\llangle A\rrangle \Box\psi}) = 1$;
		
		The above facts means $\ell(x^q_{\llangle A\rrangle\Box\psi})=1$ 
		
		\vspace*{2mm}
		\noindent iff $\ell(x^q_{\psi})=1$ and $\exists \vec{m}_A\in D_A(q):$ $\ell(e^{q,\llangle A\rrangle\Box\psi}_{A:\vec{m}_A}) = 1$, 
		
		\vspace*{2mm}
		\noindent iff $\ell(x^q_{\psi})=1$ and $\exists \vec{m}_A\in D_A(q):$($\ell(y^q_{A:\vec{m}_A})=0$ and $\ell(z^{q,\llangle A\rrangle\Box\psi}_{A:\vec{m}_A})=1$) (by lemma~\ref{lm-x16}.5) , 
		
		\vspace*{2mm}
		\noindent iff 
		$\ell(x^q_{\psi})=1$ and $\ell(x^q_{\llangle A\rrangle \bigcirc \llangle A\rrangle\Box\psi})=1$ (by the proof of the above item 4) , 
		
		\vspace*{2mm}
		\noindent iff $S\dag\eta_{\ell},q\vDash \psi$ and $S\dag\eta_{\ell},q\vDash \llangle A\rrangle \bigcirc \llangle A\rrangle\Box\psi$ (by inductive hypothesis) , 
		
		\vspace*{2mm}
		\noindent iff $S\dag\eta_{\ell},q\vDash \psi \wedge \llangle A\rrangle \bigcirc \llangle A\rrangle\Box\psi$ (by the ATL semantics) , 
		
		\vspace*{2mm}
		\noindent iff $S\dag\eta_{\ell},q\vDash\llangle A\rrangle\Box\psi$ (by lemma~\ref{lm-x16}.2) .
	\end{proof} 
	Finally, the following result follows from the above two lemmas. It shows that the proposed ILP correctly computes the allocation function.
	
	\begin{thm}\label{thm-17}~~
		\begin{itemize}
			\item[1)] If $\ell$ is a solution to {\sc ILP-Dom-SL}, then $\eta_\ell$ is a dominant social law under the bid profile $x$;
			\item[2)] If $\ell$ is a solution to {\sc ILP-Dom-IN-SL}, then $\eta_\ell$ is a dominant $(i,n)$-social law under the bid profile $x$.
		\end{itemize}	
	\end{thm}
	\begin{proof}
		1) Since $\ell$ assigns binary values to each $y^q_{i:a}$, and if $\ell$ is a solution to {\sc ILP-Dom-SL}($S, \mathcal{F}, f_1,...,f_k, x$), then it must satisfy constraint (27) which guarantees there is at least one available actions for each agent in each state, therefore there is a one-to-one correspondence between $\{\eta_\ell\}$ and social laws in $\mathcal{SL}_S$. Moreover, lemma \ref{lm-c} guarantees $\sum_{(\varphi_i,c_i)\in\mathcal{F}}c_j\cdot \ell(x^{q_s}_{\varphi_j})-\sum_{i\in Ag}\sum_{q\in Q}\sum_{a\in \epsilon_i(q)}(x_i + \frac{F_i(x_i)}{f_i(x_i)}) \ell(y^q_{i:a}) = g_(\eta_\ell,x)$.  So, this ILP basically tries to search in $\mathcal{SL}_S$ for the social law that maximizes $g(\eta,x)$. Thus, $\eta_\ell$ is a dominant social law under the bid profile $x$.
		
		\vspace*{2mm}
		2) If $\ell$ is a solution to {\sc ILP-Dominant-IN-SL($S, \mathcal{F}, f_1,...,f_k, x, i, n$)}, then the additional constraint~(60) guarantees that there is a one-to-one correspondence between $\{\eta_\ell\}$ and possible $(i,n)$-social laws, i.e., $\mathcal{SL}^{(i,n)}_S$. So, now this ILP basically tries to search in $\mathcal{SL}^{(i,n)}_S$ for the $(i,n)$-social law that maximizes $g(\eta,x)$. Therefore, $\eta_\ell$ is a dominant $(i,n)$-social law under the bid profile $x$.
	\end{proof}

	\begin{thm}\label{thm-cipl}
		Both {\sc ILP-Dom-SL} and {\sc ILP-Dom-IN-SL} can be generated in $\mathcal{O}(|Q|\cdot t\cdot l^2)$ time, where $|Q|$, $t$ are respectively the state number and state transition number of the given structure, and $l$ is the total length of the formulas in the given feature set.
	\end{thm}
	\begin{proof}
		Since for any ATL formula $\psi$, the number of formulas in $cl(\psi)$ is $\mathcal{O}(l)$(Alur, Henzinger, and Kupferman 2002), based on equation~(23), we can obtain
		\begin{equation}
			|cl(\mathcal{F})| = |cl(\varphi_1)\cup...\cup cl(\varphi_n)|\leq |cl(\bigvee_{1\leq i\leq n}\varphi_i)| = \mathcal{O}(l)
		\end{equation}
		So, $|cl(\mathcal{F})|$, formulas of the shape $p$, $\neg\varphi$, $\varphi_1\vee\varphi_2$, $\llangle A\rrangle\bigcirc\varphi$, $\llangle A\rrangle\Box\varphi$ and $\llangle A\rrangle\varphi_1\mathcal{U}\varphi_2$, and the number of mentioned coalition in $cl(\mathcal{F})$ are all $\mathcal{O}(l)$.
		Moreover, since in each state joint actions and state transitions correspond one by one, so
		\begin{equation}
			\sum_{q\in Q} |D(q)| = t
		\end{equation}
		Therefore, $|D_A(q)| = \mathcal{O}(t)$ holds for any $A\subseteq \{1,...,k\}$.
		
		According to the definition of {\sc ILP-Dom-SL}, the constraint set consists of constraints of the form (25)$\sim$ (59):    
		\begin{itemize}
			\item Constraints of the form (25) restrict the value range of all variables $x^q_{\varphi}$, and consists of $|Q|\cdot |cl(\mathcal{F})|$ constraints;
			\item Constraints of the form (26)$\sim$ (27) restrict the value range of all variables $y^q_{i:q}$, and their total number is at most 
			$$2\sum_{q\in Q} \sum_{1\leq i \leq k} |\varepsilon_i(q)|+ k|Q|\leq 2\sum_{q\in Q} \Pi_{1\leq i \leq k} |\varepsilon_i(q)|+ k|Q|$$
			\vspace*{-2mm}
			\begin{equation}
				\begin{split}
					= 2\sum_{q\in Q}|D(q)|+ k|Q| = 2t+ k|Q|
				\end{split}
			\end{equation}
			\item Constraints of the form (28)$\sim$ (30) restrict the value range of all variables $y^q_{A:\vec{m}_A}$, and their total number is at most 
			$$2\cdot|Q|\cdot 2|cl(\mathcal{F})|\cdot|D_A(q)|+k\cdot|Q|\cdot 2|cl(\mathcal{F})|\cdot|D_A(q)|$$
			\begin{equation}				
				\leq 6kt\cdot|Q|\cdot|cl(\mathcal{F})|			
			\end{equation}
			\item Constraints of the form (31)$\sim$ (34) restrict the value range of all variables $s^{q,\varphi}_{A:\vec{m}_A, \vec{m}_{\bar{A}}}$, and their total number is at most $4t\cdot |Q|\cdot |cl(\mathcal{F})|$;
			\item Constraints of the form (35)$\sim$ (37) restrict the value range of all variables $z^{q,\varphi}_{A:\vec{m}_A}$, and their total number is at most         
			\begin{equation}
				2\cdot|Q|\cdot |cl(\mathcal{F})|^2\cdot t+ |Q|\cdot |cl(\mathcal{F})|^2\cdot t = 3t\cdot|Q|\cdot|cl(\mathcal{F})|^2
			\end{equation}
			\item Constraints of the form (38)$\sim$ (39) restrict the value range of all variables $x^q_p$, and their total number is at most $|Q|\cdot |cl(\mathcal{F})|$;
			\item Constraints of the form (40) restrict the value range of all variables $x^q_{\neg \varphi}$, and their total number is at most  $|Q|\cdot |cl(\mathcal{F})|$;
			\item Constraints of the form (41)$\sim$ (43) restrict the value range of all variables $x^q_{\varphi\vee \chi}$, and their total number is at most $3\cdot |Q|\cdot |cl(\mathcal{F})|$;
			\item Constraints of the form (44)$\sim$ (47) restrict the value range of all variables $e^{q,\varphi}_{A:\vec{m}_A}$, and their total number is at most $4t\cdot |Q|\cdot |cl(\mathcal{F})|^2$;
			\item Constraints of the form (48)$\sim$ (49) restrict the value range of all variables $x^q_{\llangle A\rrangle \bigcirc\varphi}$, and their total number is at most $t\cdot |Q|\cdot |cl(\mathcal{F})| + |Q|\cdot |cl(\mathcal{F})|$;
			\item Constraints of the form (50)$\sim$ (53) restrict the value range of all variables $r^q_{\llangle A\rrangle \psi\mathcal{U}\chi}$, and their total number is at most $3\cdot |Q|\cdot |cl(\mathcal{F})| + t\cdot |Q|\cdot |cl(\mathcal{F})|$
			\item Constraints of the form (54)$\sim$ (56) restrict the value range of all variables $x^q_{\llangle A\rrangle \psi\mathcal{U}\chi}$, and their total number is at most $3\cdot |Q|\cdot |cl(\mathcal{F})|$
			\item Constraints of the form (57)$\sim$ (59) restrict the value range of all variables $x^{q}_{\llangle A\rrangle\Box\varphi}$, and their total number is at most $2\cdot |Q|\cdot |cl(\mathcal{F})| + t\cdot |Q|\cdot |cl(\mathcal{F})|$;
		\end{itemize}
		To sum up, the number of constraints is at most
		\begin{equation} 
			\begin{split}
				&11t\cdot |Q| \cdot |cl(\mathcal{F})|^2 + (6kt + 3t + 15)|Q|\cdot|cl(\mathcal{F})|\\
				&+ 2t + k|Q| + 1 \leq 39kt\cdot |Q|\cdot |cl(\mathcal{F})|^2		
			\end{split}
		\end{equation} 
		Since it takes unit time for generating each constraint and $k$ is a constant, the time required for generating {\sc ILP-Dom-SL} is $\mathcal{O}(|Q| \cdot t\cdot l^2)$. Moreover, {\sc ILP-Dominant-IN-SL} contains only 1 additional constraint of the form~(60) compared with {\sc ILP-Dom-SL}, it also takes $\mathcal{O}(|Q| \cdot t\cdot l^2)$ time to generate.
	\end{proof}

	\section{Basic Semantics of ATL}
	
	Given a {\sc ccgs} $S=\langle k, Q,  q_s, \Pi, \pi, \varepsilon, c, \delta \rangle$ Some related concepts are specified as follows:
	\begin{itemize}
		\item Basically, every agent set $A\subseteq\{1,...,k\}$ can be seen as a \emph{coalition} (with the agent set $\{1,...,k\}\setminus A$ represents the environment). The \emph{grand coalition} $\{1,...,k\}$ is denoted as $Ag$. We will sometimes use $\vec{m}$ to refer to a joint action $\langle j_1,...,j_k \rangle$ (of all the agents), use $\vec{m}_A$ to refer to a joint action of the coalition $A\subset Ag$ (called an \emph{$A$-action}) and use $D_A(q)$ to refer to all the possible $A$-actions at the state $q$.	%\footnote{For the joint actions of an arbitrary agent set $A\subseteq Ag$, we always consider them as action vectors arranged in order of increasing IDs of the corresponding agents in $A$, instead of action sets.}
		
		\item For two states $q$ and $q'$, $q'$ is called a \emph{successor} of $q$ if there is a joint action $\vec{m}\in D(q)$ such that $q'=\delta(q,\vec{m})$. 
		
		\item A \emph{computation} of $S$ is an infinite sequence $\lambda=q_0,q_1,q_2,...$ of states such that for all positions $i\geq 0$, the state $q_{i+1}$ is a successor of the state $q_i$. We refer to a computation starting from state $q$ as a $q$-computation.
		
		\item For a computation $\lambda$ and a position $i\geq 0$, we use
		$\lambda[i], \lambda[0,i], $and $\lambda[i,\infty]$ to denote,
		respectively, the $i$th state of $\lambda$, the finite prefix
		$q_0,q_1,...,q_i$ of $\lambda$, and the infinite suffix
		$q_i,q_{i+1},...$ of $\lambda$. 
		
		\item A \emph{strategy} for agent
		$a\in\Sigma$ is a function $f_a$ that maps every nonempty finite
		state sequence $\lambda\in Q^+$ to an action such that if the last
		state of $\lambda$ is $q$, then $f_a(\lambda)\in \{1,...,d_a(q)\}$.
		
		\item Finally, the \emph{outcomes} of a set of strategies $F_A$, called an
		\emph{$A$-strategy}, one for each agent in $A\subseteq Ag$, from a
		state $q\in Q$ is the set $out(q, F_A)$ of computations, such that a
		computation $\lambda=q_0,q_1,q_2,...$ is in $out(q, F_A)$ if $q_0=q$
		and there is a joint action $\langle j_1,...,j_k\rangle\in D(q_i)$
		such that (1) $j_a=f_a(\lambda[0,i])$ for all agents $a\in A$, and
		(2) $\delta(q_i,j_1,...,j_k)=q_{i+1}$.
	\end{itemize}
	
	If all the agents form a coalition and each of them specifies a strategy that gives the action choice for every possible execution history (for every state in the non-memory case), the state transition of the system will follow a single path. But if the coalition is formed by part of the agents, then they can only partially determine the action profile in the states along their execution paths. Due to the uncertainty of strategy selection of the agents outside the coalition, the possible system transition can be multiple different paths. As a matter of fact, a coalition can reliably bring about a property if and only if it can enforce a set of paths that all of which satisfy this property. The logic ATL is designed to specify and verify this kind of strategic reliability of coalitions.

	Given an ATL formula $\varphi$ and a {\sc ccgs} $S$, we use the notation $S, q\vDash \varphi$ to mean ``$\varphi$ is satisfied in the state $q$ of $S$". When $S$ is clear from the
	context, we write $q\vDash\varphi$. The relation $\vDash$ is
	defined, for all states $q$ of $S$, inductively as follows:
	\begin{itemize}
		\item For all $p\in\Pi$ we have $q\vDash p$ iff $p\in\pi(q)$.
		\item $q\vDash\neg\varphi$ iff $q\nvDash\varphi$.
		\item $q\vDash\varphi_1\vee\varphi_2$ iff $q\vDash\varphi_1$ or
		$q\vDash\varphi_2$.
		\item $q\vDash\llangle A\rrangle \bigcirc\varphi$ iff there exists
		a $A$-strategy, $F_A$, such that for all computations $\lambda\in
		out(q,F_A)$ we have $\lambda[1]\vDash\varphi$.
		\item $q\vDash\llangle A\rrangle \Box\varphi$ iff there exists
		a $A$-strategy, $F_A$, such that for all computations $\lambda\in
		out(q,F_A)$ and all positions $i\geq 0$, we have
		$\lambda[i]\vDash\varphi$.
		\item $q\vDash\llangle A\rrangle\varphi_1\mathcal{U}\varphi_2$ iff there exists
		a $A$-strategy, $F_A$, such that for all computations $\lambda\in
		out(q,F_A)$ there exists a position $i\geq 0$ such that
		$\lambda[i]\vDash\varphi_2$ and for all positions $0\leq j < i$ we
		have $\lambda[j]\vDash\varphi_1$.
	\end{itemize}
	
	Moreover, as an abbreviation, we write $\llangle A\rrangle\Diamond\varphi$ for
	$\llangle A\rrangle \top\mathcal{U}\varphi$. Based on the definition, basic syntax objects of ATL as follows intuitively capture the following meanings.
	$\llangle A\rrangle \bigcirc\varphi$ means ``\textit{coalition $A$ can reliably make $\varphi$ satisfied in the next state"};
	$\llangle A\rrangle \Box\varphi$ means ``\textit{coalition $A$ can reliably make $\varphi$ always satisfied in the subsequent states"};
	$\llangle A\rrangle\varphi_1\mathcal{U}\varphi_2$ means ``\textit{coalition $A$ can reliably make $\varphi_1$ satisfied in the subsequent states until arriving at a state satisfying $\varphi_2$"};
	$\llangle A\rrangle \Diamond\varphi$ means``\textit{coalition $A$ can reliably make $\varphi_1$ eventually satisfied in one of the subsequent states"}.

	%Algorithm 2 and algorithm 3 generate the corresponding ILPs in polynomial-time (according to theorem~\ref{thm-cipl}) and can be efficiently handled based on current ILP solvers, and according to theorem~\ref{thm-17} they correctly find out the dominant social law and dominant $(i,n)$-social law. Moreover, we have already shown that algorithm 1 correctly determines the payment to each agent via calling the above 2 algorithms. Therefore, via designing algorithm 1$\sim$3, we have proposed an effective way to compute the proposed mechanism. 

\end{document}